\documentclass[11pt]{article}

\usepackage[margin=1in]{geometry}

\usepackage{color}
\definecolor{darkgreen}{rgb}{0,0.5,0}

\usepackage[backref=page]{hyperref}
\hypersetup{
	unicode=false,          
	colorlinks=true,        
	linkcolor=red,          
	citecolor=darkgreen,        
	filecolor=magenta,      
	urlcolor=cyan           
}

\renewcommand*{\backref}[1]{}

\renewcommand*{\backrefalt}[4]{(cited on {%
		\ifcase #1 \relax Not cited.%
		\or p.~#2%
		\else pp. #2%
		\fi%
})}

\usepackage{enumitem}

\usepackage{comment}

\usepackage{amsthm, amsmath, amsfonts, amssymb}

\usepackage{authblk}

\title{Hybrid Decision Trees: Longer Quantum Time is Strictly More Powerful}
\author[1]{Xiaoming Sun \thanks{Email: \href{mailto:sunxiaoming@ict.ac.cn}{\tt sunxiaoming@ict.ac.cn}}}
\author[1]{Yufan Zheng \thanks{Email: \href{mailto:lwins.lights@gmail.com}{\tt lwins.lights@gmail.com}}}
\affil[1]{Institute of Computing Technology, Chinese Academy of Sciences, China}
\date{}

\newtheorem{theorem}{Theorem}[section]
\newtheorem{lemma}[theorem]{Lemma}
\newtheorem{meta-theorem}[theorem]{Meta-Theorem}
\newtheorem{claim}[theorem]{Claim}

\newtheorem{corollary}[theorem]{Corollary}

\newtheorem{conjecture}[theorem]{Conjecture}

\newtheorem{hypothesis}[theorem]{Hypothesis}

\newcommand{\qo}[1]{O_{#1}}
\newcommand{\oi}{\{0,1\}}
\newcommand{\bra}[1]{\langle #1 |}
\newcommand{\ket}[1]{|#1 \rangle}
\newcommand{\braket}[2]{\langle #1 | #2 \rangle}
\newcommand{\algo}{\mathcal{A}}
\newcommand{\limit}[1]{${#1}$-limited}
\newcommand{\bad}[1]{${#1}$-bad}
\newcommand{\hyb}[2]{\mathrm{Q}({#1};{#2})}
\newcommand{\qua}[1]{\mathrm{Q}({#1})}
\newcommand{\adeg}[1]{\widetilde{\deg}({#1})}
\newcommand{\cla}[1]{\mathrm{R}({#1})}
\newcommand{\dcla}[1]{\mathrm{D}({#1})}
\newcommand{\bs}[1]{\mathrm{bs}({#1})}
\newcommand{\clad}[1]{\mathrm{R}_{\rm dist}({#1})}
\newcommand{\cer}[1]{\mathrm{C}({#1})}
\newcommand{\bigo}{\mathcal{O}}
\newcommand{\tbigo}{\widetilde{\bigo}}
\newcommand{\tomega}{\widetilde{\Omega}}
\newcommand{\func}[1]{\textsc{#1}}
\newcommand{\mo}[1]{\left| {#1} \right|}
\newcommand{\one}[1]{\mathbf{1}_{#1}}
\newcommand{\inv}[2]{\beta^{(#1)}_{#2}}
\newcommand{\qm}[3]{\pi^{#1}_{#2}(#3)}
\newcommand{\res}[2]{{#1}|_{#2}}
\newcommand{\zeee}{\zeta}
\newcommand{\zee}[1]{z_{#1}}
\newcommand{\ex}[2]{x_{{#1},{#2}}}
\newcommand{\exs}{x}
\newcommand{\map}{\sigma}
\newcommand{\pot}[3]{\Psi_{#1,#2}(#3)}
\newcommand{\diff}[2]{\Delta_{#1}(#2)}
\newcommand{\ketinit}{\ket{\text{\rm init}}}
\newcommand{\pott}[1]{\tilde{\Psi}(#1)}
\newcommand{\cc}[1]{\mathsf{#1}}

\DeclareMathOperator*{\expt}{\mathbb{E}}
\DeclareMathOperator*{\var}{Var}
\DeclareMathOperator{\dist}{dist}
\DeclareMathOperator{\poly}{poly}
\DeclareMathOperator{\polylog}{polylog}
\DeclareMathOperator{\Adv}{Adv^{\pm}}

\begin{document}

	\maketitle
	
	\begin{abstract}
		In this paper, we introduce the hybrid query complexity, denoted as $\hyb{f}{q}$, which is the minimal query number needed to compute $f$, when a classical decision tree is allowed to call $q'$-query quantum subroutines for any $q'\leq q$.
		We present the following results:
		\begin{itemize}
			\item There exists a total Boolean function $f$ such that $\hyb{f}{1} = \tbigo(\cla{f}^{4/5})$.
			\item $\hyb{f}{q} = \Omega(\bs{f}/q + \sqrt{\bs{f}})$ for any Boolean function $f$;
			the lower bound is tight when $f$ is the $\func{Or}$ function.
			\item $\hyb{g \circ \func{Xor}_{C \log n}}{1} = \tomega(\sqrt{n})$ for some sufficiently large constant $C$, where $g := \func{BoolSimon}_n$ is a variant of Simon's problem.
			Note that $\qua{g\circ \func{Xor}_{C \log n}} = \bigo(\polylog n)$.
			Therefore an exponential separation is established.
			Furthermore, this open the road to prove the conjecture $\forall k,\,\hyb{g \circ \func{Xor}_{C \log^{k+1} n}}{\log^{k} n} = \tomega(\sqrt{n})$, which would imply the oracle separation $\cc{HP}(\cc{QSIZE}(n^\alpha))^\mathfrak{O} \subsetneq \cc{BQP}^\mathfrak{O}$ for any $\alpha$, where $\cc{HP}(\cc{QSIZE}(n^\alpha))$ is a complexity class that contains $\cc{BQTIME}(n^\alpha)^{\cc{BPP}}$ and $\cc{BPP}^{\cc{BQTIME}(n^\alpha)}$ in any relativized world.

		\end{itemize}
	\end{abstract}

	\section{Introduction}
	
	Assume $\cc{BPP} \neq \cc{BQP}$, as most researchers in quantum computing community conjectured.
	Then a natural question arises: will we lose computational power compared to $\cc{BQP}$, if the \emph{scale} of quantum subcircuits, whose input and output is classical, is limited, even if we are allowed to insert any $\cc{P}$ circuits?\footnote{Throughout this paper, all circuits are $\cc{P}$-uniform.}
	To be more concrete, define $\cc{HP}(\mathcal{C})$ to be the complexity class consisting of language that can be computed by connecting polynomially many circuits in class $\mathcal{C} := \mathcal{C}(n)$ with wires and polynomially many AND, OR and NOT gates, where $n$ denotes the global input size.
	The following folklore conjecture, which is first proposed formally by Jozsa~\cite{jozsa2006introduction}, says NO to the question above considering \emph{scale} as the depth of circuits.
	
	\begin{conjecture}[Josza's conjecture] \label{conj:josza}
		$\cc{BQP} = \bigcup_i \cc{HP}(\cc{QNC}^i)$.\footnote{Here, $\cc{QNC}^i$ consists of those quantum circuits of size at most $n^i$ and depth at most $\log^i n$. Note that the output of a $\cc{QNC}^i$ circuit should be classical. }
	\end{conjecture}
	
	Conjecture~\ref{conj:josza} is supported by the fact that many nontrivial $\cc{BQP}$ algorithms can be realized by $\cc{HP}( \cc{QNC}^i)$ circuits alternatively~\cite{DBLP:conf/focs/CleveW00,bremner2010classical,DBLP:conf/coco/BoulandFK18}.
	
	One motivation of this paper is to consider \emph{scale} as the size of circuits.
	Unlike the Josza's conjecture above, we believe that:
	\begin{conjecture} \label{conj:ours}
		$\cc{HP}(\cc{QSIZE}(n^{\alpha})) \subsetneq \cc{HP}( \cc{QSIZE}(n^{\beta}))$ for any $1 \leq \alpha < \beta$, where $\cc{QSIZE}(n^{\alpha})$ consists of those quantum circuits of size at most $n^{\alpha}$.
		Specifically, $\cc{HP}( \cc{QSIZE}(n^{\alpha})) \subsetneq \cc{BQP}$.
	\end{conjecture}
	
	Alas, neither \emph{refuting} Conjecture~\ref{conj:josza} nor \emph{proving} Conjecture~\ref{conj:ours} are seemingly in our reach at present, because these would separate $\cc{P}$ and $\cc{BQP}$, which would separate $\cc{P}$ and $\cc{PSPACE}$, which is a notoriously hard problem in the complexity theory.
	On the other side, the oracle separation $\bigcup_i \cc{HP}( \cc{QNC}^i)^{\mathfrak{O}} \subsetneq \cc{BQP}^{\mathfrak{O}}$ is proved very recently, independently by Chia, Chung and Lai~\cite{chia2019need} and Coudron and Menda~\cite{coudron2019computations}, which implies that any attempt to prove Conjecture~\ref{conj:josza} that relativizes is doomed to fail.
	
	While stumbling in computational complexity, the quantum computing community has made huge progress in query complexity and communication complexity, because in corresponding models one is able to show lower bounds by information-theoretic methods.
	In this paper we consider the black-box query model, or the decision tree model.
	We refer to the survey~\cite{buhrman2002complexity} for background of the classical and quantum decision tree model.

	\paragraph{From black-box to oracle separation.}
	There is a well-known technique that converts a separation in the black-box query model to an oracle separation result, which is usually called \emph{scaling everything up by an exponential}.
	This technique is attributed to Baker, Gill and Solovay~\cite{DBLP:journals/siamcomp/BakerGS75}, although they do not point it out explicitly. 
	In~\cite{DBLP:journals/siamcomp/BakerGS75}, they prove
	the seminal result $\exists \mathfrak{O},\, \cc{P}^{\mathfrak{O}} \subsetneq \cc{NP}^{\mathfrak{O}}$ by observing that $\dcla{\func{Or}} = n = \log^{\omega(1)} n$\footnote{$\dcla{f}$ denotes the deterministic query complexity of $f$.} but the query complexity of $\func{Or}$ can be decreased to $1 \leq \polylog n$ if only we introduce nondeterminism.
	Later, Furst, Saxe and Sipser~\cite{DBLP:journals/mst/FurstSS84} and Yao~\cite{DBLP:conf/focs/Yao85} formulated the idea for the polynomial hierarchy $\cc{PH}$:
	they pointed out the close relation between $\cc{PH}$ and $\cc{AC}^0$ circuits.
	For example, a $2^{\log^{\omega(1)} n}$ size lower bound of $\cc{AC}^0$ circuits computing $\func{Xor}$ (resp. $\func{Majority}$) can be scaled up to an oracle separation $ \oplus \cc{P}^{\mathfrak{O}} \not \subseteq \cc{PH}^{\mathfrak{O}}$ (resp. $ \cc{PP}^{\mathfrak{O}} \not \subseteq \cc{PH}^{\mathfrak{O}}$).
	For more details about the technique itself, see~\cite[Section 5.2]{DBLP:journals/corr/abs-quant-ph-0412143} and~\cite[Section 1.2]{DBLP:conf/stoc/Aaronson10}.
	The close relation between the black-box query model and the oracle Turing machine  model (or the circuit model with oracle gates) also motivates us to study the decision tree analog of the computation model related to $\cc{HP}(\cc{QSIZE}(q))$, besides that the analog is interesting in its own right.
	Roughly, the corresponding model is just deterministic decision trees that are allowed to call $q'$-query quantum subroutines for any $q' \leq q$, where $q$ is a parameter that may depend on the input size $n$.
	We give its formal definition in the next section.
	
	\subsection{Computation model}
	For a (possibly partial) function $f$, denote as $\cla{f}$ and $\qua{f}$ its bounded-error ($\epsilon=1/3$) classical and quantum query complexity, respectively.
	For quantum decision trees, in this paper, the query operator $\qo{x} $ associated with the ``queried'' input $x \in \oi^n$ is defined by
	$$
		\qo{x} \ket{i} \ket{b} \ket{\zee{}} = \ket{i} \ket{b \oplus x_i} \ket{\zee{}},
	$$
	where $x_i$ denotes the $i$th bit of $x$.
	We call three register query register, answer register and work register from left to right.
	We only consider real amplitudes, vectors and matrices throughout, which will not affect the power of quantum computing.
	
	\paragraph{Hybrid decision tree.}
	A \emph{hybrid decision tree} $T$ is a rooted tree equipped with some additional information.
	Each internal node of $T$, $v$, has its own small quantum decision tree $\algo_{v}$.
	Recall that a $t$-query quantum decision tree $\algo$ can be specified by unitaries $U_0, U_1, \dots, U_t$.
	The output of $\algo$ under input $x$, denoted as $\algo(x)$, is obtained by measuring the state
	$$
		U_t \qo{x} U_{t-1} \qo{x} \cdots \qo{x} U_1 \qo{x} U_0 \ketinit
	$$
	with the computational basis, where $\ketinit := \ket{0} \ket{0} \ket{0}$.
	Each possible output of $\algo_v$ is associated with one of $v$'s child \emph{exclusively}, i.e., there is a $1$--$1$ correspondence between children and possible outputs.
	Each leaf of $T$ is labeled with a value.
	Given an input $x \in \oi^n$, $T$ is evaluated as follows.
	First, evaluate $\algo_r(x)$, where $r$ is the root of $T$.
	Then, transition to the child of $r$ which is associated with the value of $\algo_r(x)$ we just get, and recursively evaluate the subtree induced by the child.
	The final output of $T$ is the value of the leaf that is reached eventually by the recursive process above.
	The cost of a path is defined as the sum of query complexity of $\algo_v$ for all node $v$ in the path.
	The \emph{cost} of $T$ is defined as the maximal cost of a path that connects the root and a leaf.
	Intuitively, it is the number of queries we need to evaluate $T$ in the worst case.
	
	\paragraph{Hybrid query complexity.}
	Given an integer $q \geq 1$ (which may vary if the input size $n$ changes), we are mainly interested in those hybrid decision trees of which each node has a quantum decision tree that makes no more than $q$ queries.
	We call these trees \limit{q}.
	For any partial Boolean function $f:D \to \oi$ where $D \subseteq \oi^n$, we say a hybrid decision tree $T$ \emph{computes} $f$, if for every input $x \in D$, $T$ outputs $f(x)$ with probability at least $2/3$.
	The \emph{hybrid query complexity}, denote as $\hyb{f}{q}$, is the minimal cost of a \limit{q}\ hybrid decision tree that computes $f$. 
	By definition it is easy to see that $\cla{f} \geq \hyb{f}{q} \geq \qua{f}$ and $\hyb{f}{\qua{f}} = \qua{f}$ for any $f$. 
	
	\subsection{Results} \label{subsec:res}
	
	Now that we have defined the hybrid decision tree model, one may ask: is the computing power of this model \emph{strictly} stronger than the classical one, or weaker than the quantum one?
	The positive answer to the first question comes immediately, because of the Forrelation function $\func{For}$, which is a partial function satisfying $\qua{\func{For}} = 1$ and $\cla{\func{For}} = \Omega({\sqrt{n}/\log n})$~\cite{aaronson2018forrelation}, which implies $\hyb{\func{For}}{1} = 1$.
	However, what if we require the function to be total?
	Theorem~\ref{thm:sep} says that the separation still exists.
	\begin{theorem} \label{thm:sep}
		There exists a total Boolean function $f$ such that $\hyb{f}{1} = \tbigo(\cla{f}^{4/5})$.
	\end{theorem}
	
	The main tool used in the proof of Theorem~\ref{thm:sep} is the \emph{cheat sheet method} introduced by Aaronson, Ben-David and Kothari~\cite{aaronson2016separations}.
	While the original cheat sheet method does not fit, we are able to slightly modify it for our need.
	We will prove the theorem in Section~\ref{sec:sep}.
	
	On the other hand, Theorem~\ref{thm:por} characterizes the hybrid query complexity of function $\func{PartialOr}$, which is function $\func{Or}$ restricted on input having at most ``$1$''.
	It gives the positive answer to the second question we asked above because $\qua{\func{PartialOr}} = \bigo(\sqrt{n})$ by Grover's search~\cite{DBLP:conf/stoc/Grover96}.
	
	\begin{theorem} \label{thm:por}
		$\hyb{\func{PartialOr}_n}{q} = \Theta(n/q + \sqrt{n})$.
	\end{theorem}
	
	Theorem~\ref{thm:por} is proved by adapting the adversary method~\cite{DBLP:journals/jcss/Ambainis02,hoyer2007negative,lee2011quantum} to our hybrid model.
	We exploit the fact that the upper bound of growth of the potential function  given by the adversary method is not tight \emph{at all steps} for $\func{PartialOr}$ in the proof.\footnote{Let $\Phi(t)$ be the potential function at step $t$. The general adversary method always give a bound that looks like $\Phi(t) = \bigo(\delta t)$ for some $\delta$, while the tight bound may be $\bigo(t^2)$ for $t<\delta$.}
	It is possible that our technique can be further adapted to show nontrivial lower bounds for other functions such as $\func{And} \circ \func{Or}$.
	The theorem will be proved in Section~\ref{sec:or}.
	
	Let $\bs{f}$ be the \emph{block sensitivity} of (possibly partial) function $f$~\cite{buhrman2002complexity}.
	By a well-known reduction that reduces $\func{PartialOr}_{\bs{f}}$ to $f$ (see e.g.~\cite[Section 3.4]{DBLP:journals/cc/NisanS94}), Theorem~\ref{thm:por} implies the following corollary, which gives good bounds for all Boolean function.
	
	\begin{corollary}
		$\hyb{f}{q} = \Omega(\bs{f}/q + \sqrt{\bs{f}})$ for any Boolean function $f$.
	\end{corollary}
	
	Theorem~\ref{thm:por} is not enough if we want to prove Conjecture~\ref{conj:ours} \emph{relativized to some oracle}.
	For that, we need to find a function $f$ for every integer $k \geq 1$, such that $\hyb{f}{\log^k n} = \log^{\omega(1)} n$, whereas $\qua{f} \leq \polylog n$.
	If there is such an $f$, scaling every up we would get $\cc{HP}(\cc{QSIZE}(n^k))^{\mathfrak{O}_k} \subsetneq \cc{BQP}^{\mathfrak{O}_k} $ for some oracle ${\mathfrak{O}_k}$.
	By merging oracles one get $\cc{HP}(\cc{QSIZE}(n^k))^{\mathfrak{O}} \subsetneq \cc{BQP}^{\mathfrak{O}} $ for some oracle ${\mathfrak{O}}$ and every integer $k \geq 1$.
	Therefore, Conjecture~\ref{conj:ours} follows relativized to ${\mathfrak{O}}$, by a padding argument (see e.g.~\cite[Proof of Theorem 1]{DBLP:journals/jcss/Cook73}).
	We conjecture that one can let $f_n := \func{For}_m \circ \func{Xor}_{C \log^{k+1}m}$, where $\func{Xor}$ denotes the parity function and $C$ is some sufficiently large constant.
	Obviously, $\qua{f} \leq \polylog n$.
	We hope $\hyb{f}{\log^k} = \log^{\omega(1)} n$, which is a corollary of the following conjecture that is more general.
	
	\begin{conjecture} \label{conj:main}
		There exists a universal constant $C$ such that
		$$
		\hyb{f_n \circ \func{Xor}_{C q \log n}}{q} = \Omega(\cla{f_n \circ \func{Xor}_{C q \log n}}) = \Omega(\cla{f_n} \cdot q \log n),
		$$
		where $f$ is any partial Boolean function.\footnote{The second equality holds because $\cla{f \circ \func{Xor}_m} = \Omega(m \cla{f})$ by a proof similar to that of $\cla{f \circ \func{Or}_m} = \Omega(m \cla{f})$ in~\cite{aaronson2016separations}.}
	\end{conjecture}
	
	Unfortunately, we are unable to prove this conjecture.
	Nevertheless, we can handle the case $f = \func{BoolSimon}$ and $q = 1$ in Conjecture~\ref{conj:main}.
	
	\paragraph{Simon's problem.}
	The celebrated Simon's problem~\cite{DBLP:journals/siamcomp/Simon97a} can be abstracted as a partial function, as follows.
	Let $k$ be an integer and $K = 2^k$.
	For each string $x = x_0 x_1 \dots x_{K-1}$ with $x_i \in \{0,\dots,K-1\}$ such that $\forall i \neq j,\ x_i = x_j \iff i \oplus j = a$ holds for some $a$, define $\func{Simon}(x) = a$.
	It is implicit in~\cite{DBLP:journals/siamcomp/Simon97a} that $\qua{\func{Simon}} = \bigo(k)$ and $\cla{\func{Simon}} = \Omega(\sqrt{K})$.
	Define $\func{LSBSimon}(x) = \func{Simon}(x) \bmod 2$.
	Let $\func{BoolSimon}_n$ be the Boolean version of $\func{LSBSimon}$.
	Its input should be a binary string of length $n = Kk$.
	We will prove that $\qua{\func{BoolSimon}} = \bigo(\log^2 n)$ and $\cla{\func{BoolSimon}} = \tomega(\sqrt{n})$ in Section~\ref{subsec:simon.proof}.
	
	\begin{theorem} \label{thm:simon}
		$\hyb{\func{BoolSimon}_n \circ \func{Xor}_{C \log n}}{1} = \tomega(\sqrt{n})$ for some sufficiently large constant $C$.
	\end{theorem}
	
	Note that $\qua{\func{BoolSimon}_n \circ \func{Xor}_{C \log n}} = \bigo(\log^3 n)$.
	Therefore it is a function that \emph{exponentially} separates the quantum query complexity from the hybrid query complexity.
	
	Roughly speaking, the proof strategy for Theorem~\ref{thm:simon} is to force the hybrid decision tree computing the function to degenerate to a classical decision tree.
	The idea is detailed in Section~\ref{subsec:simon.idea} along with a warm-up case.
	The formal proof of Theorem~\ref{thm:simon} lies in Section~\ref{subsec:simon.proof}.
	Finally, we discuss the possibility of proving Conjecture~\ref{conj:main} by generalizing the proof, in Section~\ref{subsec:simon.disc}.
	
	\subsection{Notation and convention}
	A (partial) function $f:D \to \oi$ with $D \subseteq \oi^n$ is called \emph{Boolean}.
	If $D = \oi^n$ we call $f$ a \emph{total} function.
	Sometimes when we refer to \emph{a function} $f$ we actually mean a function family $\{f_n\}_{n \geq 1}$.
	As shown, we will always use subscript ``$n$'' to emphasize the input size of a function $f$ if needed.
	For any two (possibly partial) Boolean functions $f_n$ and $g_m$, define $f \circ g$ by
	$$
	f \circ g(x) = f \circ g(x_{1,1},\dots,x_{n,m}) = f(y_1,\dots,y_n),\quad \forall i,\, y_i = g(x_{i,1},\dots,x_{i,m})
	$$
	for every $x$ that lets the above definition make sense.
	
	The subscript $C$ in $\bigo_C(\cdot)$ means that the constant factor hidden by $\bigo$ depends on $C$.
	We also use subscript this way in other asymptotic notations.
	Notation $\poly(n)$ and $\polylog(n)$ means $\bigo(n^{\bigo(1)})$ and $\bigo(\log^{\bigo(1)} n)$, respectively.
	$\tbigo(g)$ is the shorthand for $\bigo(g \polylog g)$ and $\tomega(g)$ is for $\Omega(g /\polylog g)$.
	Functions $\func{And}$, $\func{Or}$ and $\func{Xor}$ are defined by $\func{And}(x) = \bigwedge_{1 \leq i \leq n} x_i$, $\func{Or}(x) = \bigvee_{1 \leq i \leq n} x_i$ and $\func{Xor}(x) = \bigoplus_{1 \leq i \leq n} x_i$.
	Define $\one{P} = 1$ if proposition $P$ holds and $\one{P} = 0$ if not.
	Define $\expt_{P}[X] = \sum_{P}[X]/\sum_{P}[1]$, $\Pr_{P}[Q] = \expt_{P}[\one{Q}]$ and $\var_{P}[X] = \expt_{P}[X^2] - \expt_{P}[X]^2$.
	For matrices $A$ and $B$, $A \succeq B$ means that $A-B$ is positive semidefinite.
	Relation $\preceq$ is defined similarly.
	
	The terms decision trees and algorithms are used interchangeably.
	We ignore problems caused by non-integer throughout (e.g. some integer parameter is $\log n$, which is not guaranteed to be an integer), because they can be easily handled by truncating or other trivial techniques.
	
	\section{Separation between hybrid and classical query complexity} \label{sec:sep}
	
	In Section~\ref{subsec:sep.idea} we explain the main idea of the proof of Theorem~\ref{thm:sep}. 
	The idea is generally the same as in~\cite{aaronson2016separations}, with some modifications.
	We prove the theorem in Section~\ref{subsec:sep.proof}.
	
	\subsection{Idea} \label{subsec:sep.idea}
	
	We first introduce some functions to be used later.
	Function $\func{And-Or}_n$ is defined as $\func{And}_{\sqrt{n}} \circ \func{Or}_{\sqrt{n}}$.
	The reason we look at this function is that it achieves the best separation between the certificate complexity and the classical query complexity, i.e., $\cer{\func{And-Or}_n} = \bigo(\sqrt{n})$ and $\cla{\func{And-Or}_n} = \Omega({n})$, where $\cer{f}$ denotes the certificate complexity of $f$~\cite{buhrman2002complexity}.
	The Forrelation function, denoted as $\func{For}$, is a partial Boolean function that satisfies $\qua{\func{For}_n} = 1$ and $\cla{\func{For}_n} = \Omega({\sqrt{n}/\log n})$~\cite{aaronson2018forrelation}.
	
	Next we describe an operator that slightly generalizes the composition operator ``$\circ$''.
	Let $f:\oi^n \to \oi$ be a total function and $g$ be a partial Boolean function with input size $m$.
	The function $h = f \odot g$, with input size $nm$, is defined as follows.
	For every $b$-certificate $c \in \{0,1,\ast\}^n$ of $f$, let $f(x_{11} x_{12} \dots x_{nm}) = b$ for every $nm$-bit input $x$ that satisfies $g(x_{i1}x_{i2} \dots x_{im}) = c_i$ for every $c_i \in \oi$.
	It is easy to see that $f \odot g$ contains $f \circ g$ as a subfunction.
	
	Let $g := \func{And-Or}_{m^2}$.
	Now consider the partial function $f := g \odot \func{For}_m$.
	It is straightforward that $\hyb{g \circ \func{For}_m}{1} = \bigo(m^2 \log m)$: For each input of $g$, we run the algorithm that computes the corresponding $\func{For}_m$ (with probability at least $2/3$) $10 \log m$ times to amplify the success probability to at least $1 - 1/m^3$.
	Then, by a union bound we are able to compute $g \circ \func{For}_m$ classically with probability at least $2/3$.
	The same algorithm can be also perfectly applied to $f$.
	Therefore $\hyb{f}{1} = \bigo(m^2 \log m)$.
	As for $\cla{f}$, intuitively it should be the case that 
	$$
		\cla{f} \geq \cla{g \circ \func{For}_m} = \Omega(\cla{g} \cdot \cla{\func{For}_m}) = \Omega(m^{2.5}/\log m).
	$$ 
	However, the first equality relies on a general composition theorem that claims $\cla{p \circ q} = \Omega(\cla{p} \cdot \cla{q})$, which is unknown to be true at present.
	Nevertheless, we will fix this issue later by applying a slightly weaker version of composition theorem proved by Ben-David and Kothari~\cite{ben2016randomized}.
	
	Our construction using ``$\odot$'' gives $f$ a nice property: An input $x$ is in the domain of $f$ if and only if there exist $k = \bigo(m^2)$ positions $i_1, \dots, i_k$,  such that one can verify that $x$ is indeed in the domain, given only $x_{i_1}, \dots, x_{i_k}$.
	We call $(i_1, \dots, i_k)$ a \emph{certificate} of the fact that $x$ is in the domain.
	
	Set $n = m^3$ to be the input length of $f$.
	Let $f_{\rm CS}$ be the cheat sheet version of $f$ that consists of $10 \log n$ copies of $f$, taking input $(x^{(1)}, \dots, x^{(10\log n)}, Y_0, \dots, Y_{n^{10}-1})$, where $x^{(i)}$ is of length $n$ and $Y_i$ is of length $\bigo(m^2 \log^2 n)$.
	Recall from~\cite{aaronson2016separations} that $f_{\rm CS}(x^{(1)}, \dots, x^{(10\log n)}, Y_0, \dots, Y_{n^{10}-1}) = 1$ if and only if the following condition holds:
	\begin{enumerate}[label=(\roman*)]
		\item $x^{(i)}$ is in the domain of $f$ for $i = 1,\dots,10 \log n$.
		\item Let $l_i := f(x^{(i)})$ and $l = \sum_{1 \leq i \leq 10 \log n} l_i \cdot 2^{i-1}$. 
		$Y_l$ consists of certificates of the fact that $x^{(i)}$ is in the domain of $f$, for $i = 1,\dots,10 \log n$.\footnote{Condition (ii) is seemingly slightly different to the one in~\cite{aaronson2016separations}. However, two definitions are the same for our specific $f$.}
	\end{enumerate}
	Now let us examine the query complexity of $f_{\rm CS}$.
	For the hybrid one we have the following algorithm.
	First, feed $x^{(i)}$ into the algorithm that computes $f$ and get the output $l_i$, for each $i$, regardless of the fact that $x^{(i)}$ may not be in the domain of $f$.
	Then, classically query entire $Y_l$, where $l = \sum_{1 \leq i \leq 10 \log n} l_i \cdot 2^{i-1}$ and $l_1,\dots,l_{10 \log n}$.
	The algorithm returns $1$ if and only if it verifies that the condition (ii) above indeed holds.
	Therefore, $\hyb{f_{\rm CS}}{1} = \tbigo(m^2)$.
	On the other hand, intuitively the only efficient way to find the correct $Y_l$ is to compute $f(x^{(i)})$ for every $i$.
	So we expect $\cla{f_{\rm CS}} = \tomega({m^{{2.5}}})$.
	Hence we get the separation desired.
	
	\subsection{Proof of Theorem~\ref{thm:sep}}  \label{subsec:sep.proof}
	
	We plan to use the following theorem instead of the unproven composition theorem of classical query complexity.
	\begin{theorem}[Ben-David and Kothari~\cite{ben2016randomized}] \label{thm:BK}
		Let $f$ and $g$ be partial Boolean functions and let $\func{Ind}$ be the index function\footnote{In some literature, $\func{Ind}$ is called the address function. We omit the definition here since it is not relevant for our purpose.}.
		Let $m \in \Omega(\cla{g}^{{1.1}}) \cap \poly(\cla{g})$, which implies $\cla{\func{Ind}_m} = \Theta(\log \cla{g})$.
		Then,
		$$
			\cla{f \circ \func{Ind}_m \circ g} = \Omega(\cla{f} \cla{\func{Ind}_m} \cla{g}) = \Omega(\cla{f} \cla{g} \log \cla{g}).
		$$
	\end{theorem}
	Recall that $\cla{\func{Ind}_n} = \cer{\func{Ind}_n} = \Theta(\log n)$.
	Analogous to the idea in the previous section, we define the partial function $f = g \odot \func{For}_m$ with $g := \func{And-Or}_{m^2} \circ \func{Ind}_{m^2}$.
	We have $\cla{g} = \tbigo(m^2)$ and $\cer{g} = \tbigo(m)$.
	Set $n = m^5$ to be the input size of $f$.
	We still have $\hyb{f}{1} = \tbigo(m^2)$.
	Since $f$ contains $\func{And-Or}_{m^2} \circ \func{Ind}_{m^2} \circ \func{For}_m$ as a subfunction, applying Theorem~\ref{thm:BK} we get $\cla{f} = \tomega(m^{2.5})$.
	For the cheat sheet version of $f$, $\hyb{f_{\rm CS}}{1} = \tbigo(m^2)$ still holds via the algorithm we have developed.
	On the other hand, we can apply a general theorem on the query complexity of $f_{\rm CS}$ from Aaronson, Ben-David and Kothari~\cite{aaronson2016separations}, which states that $\cla{f_{\rm CS}} = 
	\Omega(\cla{f}/ \log n)$.
	Therefore, $\cla{f_{\rm CS}} = \tomega(m^{2.5})$.
	
	\section{Hybrid query complexity of the OR function} \label{sec:or}
	
	In Section~\ref{subsec:or.idea} we explain the main idea of the proof of Theorem~\ref{thm:por}. 
	Then we prove the theorem in Section~\ref{subsec:or.proof}.
	
	\subsection{Idea} \label{subsec:or.idea}
	
	The upper bound $\hyb{\func{PartialOr}}{q} = \bigo(n/q + \sqrt{n})$ is very easy to show:
	since $\qua{\func{PartialOr}} = \bigo({\sqrt{n}})$, we only need to consider the case $q \leq \sqrt{n}$.
	For that, one can divide input into $n/q^2$ blocks of $q^2$ bits and then resort to Grover's search~\cite{DBLP:conf/stoc/Grover96} to check whether there is an ``$1$'' in each block using $n/q^2 \cdot \bigo(\sqrt{q^2}) = \bigo(n/q)$ queries.
	
	For technical convenience\footnote{Function $\func{Find}$ allows us to write its potential function from the adversary method in a more ``symmetric'' form.}  we introduce a new partial Boolean function $\func{Find}$.
	Its valid inputs are strings of Hamming weight $1$, i.e., the string that contains exactly one ``$1$''.
	Its value is $1$ if and only if the ``$1$'' appears in the first half of the input. A simple reduction shows that:
	\begin{lemma} \label{lem:tri}
		$\hyb{\func{PartialOr}_n}{q} = \Omega(\hyb{\func{Find}_n}{q})$.
	\end{lemma}
	Below we will focus on proving $\hyb{\func{Find}_n}{q} = \Omega(n/q + \sqrt{n})$.
	To ease the lower bound analysis, we say a \limit{q}\ hybrid decision tree $T$ \emph{regular} if 
		every internal node of $T$ has a $q$-query quantum decision tree.
	Note that a \limit{q}\ tree $T$ of cost $c$ can be simulated by a regular \limit{ q}\ tree $T'$ of depth at most $2\lceil c/q \rceil$, via a simple greedy argument.
	Therefore, to prove $\hyb{f}{q} \geq c$, it suffices to show a $2\lceil c/q \rceil$ depth lower bound for any regular \limit{q}\ tree that computes $f$.
	
	To better understand why $\hyb{\func{Find}}{q}$ is large when $q$ is small (i.e. $q \ll \sqrt{n}$), it is beneficial to consider what a $q$-query quantum algorithm can do to help computing $\func{Find}$.
	Recall how the adversary method~\cite{DBLP:journals/jcss/Ambainis02,hoyer2007negative,lee2011quantum} works to lower bound $\qua{\func{Find}}$:
	Denote $e_i$ as the input in which ``$1$'' is in the $i$th position.
	Suppose the quantum decision tree for $\func{Find}$ is specified by unitaries $U_0, U_1, \dots, U_k$.
	Define $\ket{\psi_i^{(t)}} = U_t \qo{e_i} U_{t-1} \qo{e_i} \cdots \qo{e_i} U_1 \qo{e_i} U_0 \ketinit$.
	Choosing an appropriate adversary matrix gives us the potential function
	$$
		\Phi(t) = \sum_{\substack{1 \leq i \leq n/2\\ n/2 < j \leq n}} \left(1 - \braket{\psi_i^{(t)}}{\psi_j^{(t)}} \right). 
	$$
	We modify it to be more symmetric to handle. Let
	$$
	\Phi(t) := \frac{1}{n^2} \sum_{\substack{1 \leq i,j \leq n}} \left(1 - \braket{\psi_i^{(t)}}{\psi_j^{(t)}} \right) = \frac{1}{2n^2}  \sum_{\substack{1 \leq i,j \leq n}} \mo{\ket{\psi_i^{(t)}} - \ket{\psi_j^{(t)}}}^2.
	$$
	Clearly $\Phi(0) = 0$ and $\Phi(k) = \Omega(1)$ since the algorithm computes $f$ with probability at least $2/3$.
	A standard argument then follows to show that $\Phi(t) - \Phi(t-1) = \bigo(1/\sqrt{n})$, which implies $k = \Omega(\sqrt{n})$.
	
	Now, for any $q$-query algorithm, the above argument gives $\Phi(q) = \bigo(q /\sqrt{n})$.
	However this bound is far from being tight.
	In fact, for this specific potential function, it is known that $\Phi(q) = \bigo(q^2/n)$~\cite{boyer1998tight}. 
	This enlighten us to design a potential function $\Psi(h)$ analogously for the hybrid decision tree.
	We expect $\Psi(h)$ to represent the ``progress'' the tree already made, when one reaches a depth-$h$ node while evaluating the tree.
	We hope that $\Psi(h) - \Psi(h-1) = \bigo(q^2/n)$, i.e., the quantum ``progress'' accumulates \emph{additively}, from which we would obtain $h = \Omega(n/q^2)$ if $\Psi(h) = \Omega(1)$.
	Thus, the regular \limit{q}\ hybrid decision tree computing $f$ should have depth at least $\Omega(n/q^2)$.
	
	Arguably, the most natural way to define $\Psi$ is as follows.
	Suppose we have a regular \limit{q}\ tree $T$.
	\newcommand{\pb}[2]{\alpha^{(#1)}_{#2}}
	Denote $\pb{v}{i}$ as \emph{square root of} the probability that node $v$ is reached if the input is $e_i$ when evaluating $T$.
	Let
	$$
		\Psi(h) := \frac{1}{2n^2} \sum_{\substack{v :\,\text{$v$ has depth $h$} \\ 1 \leq i,j \leq n}} \left( \pb{v}{i} - \pb{v}{j} \right)^2.
	$$
	Unfortunately, $\Psi(h) - \Psi(h-1)$ may as large as $\Theta(q/\sqrt{n})$ with this definition.
	However, if for every node $v$, $\pb{v}{1}, \dots, \pb{v}{n}$ are in $\bigo( \expt_{1 \leq i \leq n} [ \pb{v}{i} ] )$, then we do have $\Psi(h) - \Psi(h-1) = \bigo(q^2/n)$.
	While this condition may not hold for the original decision tree, we can modify the tree such that all nodes concerned satisfy the condition, in a way that will not decrease the success probability of computing $\func{Find}$.
	More concretely, if there is a node $v$ in the original tree and some $j$ such that $\pb{v}{j} \gg  \expt_{1 \leq i \leq n} [ \pb{v}{i} ] $, we insert a node $u$ between $v$ and its parent.
	When the input is $e_j$, $u$ transitions to a leaf that is labeled with value $\func{Find}(e_j)$.
	Otherwise, it transitions to $v$.
	The only remaining work is to upper bound the number of insertion so that we are guaranteed that most ``work'' (i.e. growth of $\Psi$) of computing $\func{Find}$ is done by the original nodes but not nodes we inserted.
	
	
	\subsection{Proof of Theorem~\ref{thm:por}} \label{subsec:or.proof}
	
	The upper bound is obtained by the algorithm discussed in the previous section.
	We prove $\hyb{\func{Find}}{q} = \Omega(n/q + \sqrt{n})$ for $q \leq \sqrt{n}$ below.
	This suffices due to Lemma~\ref{lem:tri} and $\hyb{\func{Find}}{q} \geq \qua{\func{Find}} = \Omega(\sqrt{n})$.
	
	Now suppose we have a regular \limit{q}\ hybrid decision tree $T$ of depth $h$ that computes $\func{Find}$.
	\newcommand{\con}{C}
	Let $\con>2$ be a constant.
	
	\paragraph{Tree surgery.}
	We repeatedly apply the following process to $T$ whenever there is an \emph{original internal} node $v$ and some $j$ such that $\pb{v}{j} \geq \con \cdot \expt_{1 \leq i \leq n} [ \pb{v}{i} ]$.
	Let $w$ be $v$'s parent.
	We modify $w$ such that it will transition to a new node $u$, if $\algo_w(x)$ returns a value that originally associated with $v$ when evaluating $T$.
	We define $\algo_u$ such that $u$ transitions to $v$ if $x \neq e_j$ and transitions to a new leaf node $\ell$ labeled with value $\func{Find}(e_j)$ if $x = e_j$.
	We call $u$ and $\ell$ inserted this way \emph{auxiliary} nodes.
	
	\paragraph{Potential function.}
	Let $S$ be a set of nodes in $T$ such that any path connecting the root and a leaf passes through exactly one node in $S$.
	We call $S$ a \emph{cut} of $T$.
	Define
	\begin{align*}
		\Psi(S) &= \Psi_1(S) + \Psi_2(S) \\
		\Psi_1(S) &= \sum_{v \in S} \var_{1 \leq i \leq n}[\pb{v}{i}] \\
		\Psi_2(S) &= \frac{2}{\con-2} \cdot \sum_{\substack{v \in S \\ \text{$v$ is not an auxiliary leaf}}} \var_{1 \leq i \leq n}[\pb{v}{i}]
	\end{align*}
	Let $S_{\rm leaf}$ be the collection of all leaves of $T$.
	Note that $$\Psi_1(S) = \dfrac{1}{2n^2} \sum_{\substack{v \in S \\ 1 \leq i,j \leq n}} \left( \pb{v}{i} - \pb{v}{j} \right)^2 \geq \dfrac{1}{2n^2} \sum_{\substack{v \in S \\ 1 \leq i \leq n/2 \\ n/2<j \leq n}} \left( \pb{v}{i} - \pb{v}{j} \right)^2.$$
	View $T$ as a pure quantum decision tree.
	Since $T$ computes $\func{Find}$ with probability at least $2/3$, we obtain $\Psi_1(S_{\rm leaf}) \geq \frac{1}{2n^2} \Omega(n^2) = \Omega(1)$ via the output condition for the adversary method~\cite{hoyer2007negative,lee2011quantum}.
	
	\paragraph{Growth of $\Psi$.}
	We dynamically maintain a cut $S$ and ``push'' it down step by step to $S_{\rm leaf}$.
	During this process we bound the growth of $\Psi(S)$.
	Initially, let $S \leftarrow \{r\}$ where $r$ is the root.
	Clearly $\Psi(S) = 0$.
	We repeatedly update $S$ using the following rules:
	\begin{itemize}
		\item {\bf Rule 1.} If there is an internal auxiliary node $v \in S$, delete $v$ from $S$ and add all its children into $S$.
		\item {\bf Rule 2.} If $S$ does not contain any internal auxiliary node. Replace $S$ by the set 
		$$
			\{v \in S : v \text{ is a leaf} \} \cup \{w : w \text{ is a child of an internal node $v \in S$} \}.
		$$
	\end{itemize}
	Since the original tree is of depth $h$,
	we will get $S = S_{\rm leaf}$ after applying Rule 1 several times and Rule 2 exactly $h$ times.
	Then, Lemma \ref{lem:rule1} and \ref{lem:rule2} below give $\Psi(S_{\rm leaf}) = \bigo_{\con}(q^2 h/n)$.
	Combine this with $\Psi(S_{\rm leaf}) \geq \Psi_1(S_{\rm leaf}) = \Omega(1)$ we obtain $h = \Omega_{\con}(n/q^2)$, which implies $\hyb{\func{Find}}{q} = \Omega_{\con}(n/q)$, as desired.
	
	\begin{lemma} \label{lem:rule1}
		Rule 1 does not increase $\Psi(S)$.
	\end{lemma}
	
	\begin{lemma} \label{lem:rule2}
		Rule 2 increases $\Psi(S)$ by at most $\left(1+\frac{2}{\con -2} \right) 4C^2 \cdot (q^2/n)$.
	\end{lemma}
	
	\begin{proof}[Proof of Lemma~\ref{lem:rule1}]
		Denote $\delta_1$ and $\delta_2$ as the increment of $\Psi_1(S)$ and $\Psi_2(S)$ after applying Rule 1, respectively.
		Denote as $v$ the auxiliary node that will be deleted.
		Let $j$ be the one that satisfies $\algo_v(e_j) \neq \algo_v(e_i)$ for any $i \neq j$.
		By definition of auxiliary nodes $\algo_v(e_j)$ is associated with a leaf $\ell$, while $\algo_v(e_i)$ for $i\neq j$ are associated with another node $u$.
		We have $\pb{\ell}{i} = \pb{v}{i} \cdot \one{i=j}$ and $\pb{u}{i} = \pb{v}{i} \cdot \one{i\neq j}$.
		Therefore,
		\begin{align*}
			\delta_1 &= \var_{1 \leq i \leq n} [\pb{u}{i}]+\var_{1 \leq i \leq n} [\pb{\ell}{i}]- \var_{1 \leq i \leq n} [\pb{v}{i}] \\
			&= \expt_{1 \leq i \leq n} [\pb{v}{i}] - \expt_{1 \leq i \leq n} [\pb{u}{i}]- \expt_{1 \leq i \leq n} [\pb{\ell}{i}] \\
			&= \mu^2 - (\mu-\lambda)^2 - \lambda^2 & \text{Let $\mu = \expt_{1 \leq i \leq n} [\pb{v}{i}]$ and $\lambda = \frac{\pb{v}{j}}{n}$} \\
			&= 2\lambda \mu - 2 \lambda^2 \\
			& \leq 2 \lambda \mu. \\
			\delta_2 &= \frac{2}{\con-2} \left( \var_{1 \leq i \leq n} [\pb{u}{i}] - \var_{1 \leq i \leq n} [\pb{v}{i}] \right) \\
			&= \frac{2}{\con-2} \left( 2\lambda \mu - 2 \lambda^2 - \var_{1 \leq i \leq n} [\pb{\ell}{i}]  \right) \\
			&= \frac{2}{\con-2} \left( 2\lambda \mu - 2 \lambda^2 - \lambda(n-1) \right) \\
			&\leq -2 \lambda \mu. & \lambda \geq \frac{\con}{n} \mu
		\end{align*}
	\end{proof}
	
	\begin{proof}[Proof of Lemma~\ref{lem:rule2}]
		Let $S_0$ be $S$ before applying the rule and $v$ be an internal node in $S_0$.
		The increment $v$ contributes to $\Psi_1(S)$ is
		\begin{align*}
			\sum_{u :\, \text{child of $v$}} \var_{1 \leq i \leq n} [\pb{u}{i}] - \var_{1 \leq i \leq n} [\pb{v}{i}] = \frac{1}{2n^2} \sum_{1\leq i,j \leq n} \left( \sum_u \left(\pb{u}{i}-\pb{u}{j}\right)^2 - \left(\pb{v}{i}-\pb{v}{j}\right)^2 \right).
		\end{align*}
		For the $q$-query quantum decision tree $v$ has, define $\ket{\psi_i}$ as the final state when the input is $e_i$.
		Note that the vector $\left( \pb{u}{i} \right)_{u :\, \text{child of $v$}}$ is exactly $\pb{v}{i} \ket{\psi_i}$ with each component taking the absolute value.
		So the right-hand side of above equality is no more than
		\begin{align*}
			& \frac{1}{2n^2} \sum_{1\leq i,j \leq n} \left( \mo{\pb{v}{i}\ket{\psi_i} - \pb{v}{j}\ket{\psi_j} }^2 - \left(\pb{v}{i}-\pb{v}{j}\right)^2 \right) \\
			=& \frac{1}{2n^2} \sum_{1\leq i,j \leq n} \pb{v}{i} \pb{v}{j} \mo{\ket{\psi_i} - \ket{\psi_j} }^2 \\
			\leq & \frac{1}{2n^2} \cdot \left( \con \expt_{1 \leq i \leq n}[\pb{v}{i}] \right)^2 \sum_{1\leq i,j \leq n} \mo{\ket{\psi_i} - \ket{\psi_j} }^2 \\
			\leq & \frac{C^2}{2n^2} \cdot  \expt_{1 \leq i \leq n}[{\pb{v}{i}}^2] \sum_{1\leq i,j \leq n} \mo{\ket{\psi_i} - \ket{\psi_j} }^2 \\ 
			\leq & \frac{C^2}{2n^2} \cdot  \expt_{1 \leq i \leq n}[{\pb{v}{i}}^2]  \cdot 8q^2 n & (\ast) \\
			= & \frac{4C^2q^2}{n} \expt_{1 \leq i \leq n}[{\pb{v}{i}}^2], &
		\end{align*}
		where $(\ast)$ holds because $\sum_{1\leq i,j \leq n} \mo{\ket{\psi_i} - \ket{\psi_j} }^2 \leq 2n\sum_{1\leq i \leq n} \mo{\ket{\psi_i} - \ket{\phi} }^2 \text{ for any $\ket{\phi}$}$ and we have the following claim.
		\begin{claim}[Boyer et al.~\cite{boyer1998tight}]
			For any $q$-query quantum algorithm, denote $\ket{\psi_0}$ as the final state when input $x = \boldsymbol{0} := 0^n$ and $\ket{\psi_i}$ as the final state when $x = e_i$, where $e_i$ is $\boldsymbol{0}$ with $i$th bit flipped.
			Then,
			$$
			\sum_{1 \leq i \leq n} \mo{\ket{\psi_i} - \ket{\psi_0}}^2 \leq 4q^2.
			$$
		\end{claim}
		
		Therefore, the total increment of $\Psi_1(S)$ must be at most 
		$$
			\sum_{v \in S_0} \frac{4C^2q^2}{n} \expt_{1 \leq i \leq n}[{\pb{v}{i}}^2] = \frac{4C^2q^2}{n} \expt_{1 \leq i \leq n}[\sum_{v \in S_0} {\pb{v}{i}}^2] = \frac{4C^2q^2}{n}.
		$$
		Doing the similar analysis for $\Psi_2(S)$, we can upper bound the increment of $\Psi(S)$ by $\left(1+\frac{2}{\con -2} \right) \cdot \frac{4C^2q^2}{n}$.
	\end{proof}
	
	\section{Separation between hybrid and quantum query complexity} \label{sec:simon}
	
	In Section~\ref{subsec:simon.idea} we explain the main idea of the proof of Theorem~\ref{thm:simon}. 
	Then we prove the theorem in Section~\ref{subsec:simon.proof}.
	Finally, we discuss the extendability of the proof, in Section~\ref{subsec:simon.disc}.
	
	\subsection{Idea} \label{subsec:simon.idea}
	
	Let us consider the more general Conjecture~\ref{conj:main} instead of Theorem~\ref{thm:simon}.
	Recall from Section~\ref{subsec:or.idea} that we can assume that the \limit{q}\ hybrid decision tree $T$ computing $f_n \circ \func{Xor}_{C q \log n}$ is regular, without loss of generality.
	We introduce some symbols and definitions for further discussion.
	Let $m := Cq \log n$.
	We usually use $\zee{} = \zee{1} \zee{2} \dots \zee{n}$ to denote an input (which may be illegal) of $f_n$.
	Similarly $\exs = \ex{1}{1} \dots \ex{1}{m} \ex{2}{1} \dots \ex{2}{m} \dots \ex{n}{m}$ is an input of $f_n \circ \func{Xor}_m$.
	We call $\ex{i}{1},\dots,\ex{i}{m}$ the \emph{$i$th block} of $\exs$.
	Define function $\map:\oi^{nm} \to \oi^n$ by $\map(\exs) = \zee{}$ with $\zee{i} = \bigoplus_{1 \leq j \leq m} \ex{i}{j}$ for $i=1,\dots,n$.
	Clearly $f_n \circ \func{Xor}_{m}(\exs) = f_n(\map(\exs))$.
	Denote $\pb{v}{\exs}$ as the probability that node $v$ is reached if the input is $\exs$ when evaluating $T$.\footnote{Note that this definition is different from the one in Section~\ref{subsec:or.idea}.}
	Let $\pb{v}{\zee{}} := \expt_{\exs : \map(\exs) = \zee{}}[\pb{v}{\exs}]$.
	
	\paragraph{Intuition.}
	Conjecture~\ref{conj:main} states that the ability to run quantum algorithms in $T$ essentially provides no additional benefit compared to just computing classically.
	Thus, $T$ somewhat ``degenerates'' to a classical decision tree.
	The following natural hypothesis will result in this degeneration.
	\begin{hypothesis} \label{hypo:main}
		One can assign values $\inv{v}{1},\dots,\inv{v}{n}$ to each internal node $v$ in $T$  such that $\sum_{1 \leq i \leq n} \inv{v}{i} = q$.
		There exists an assignment such that for any node $v$, $z$ and $z'$ being legal inputs of $f$, we have $\pb{v}{\zee{}} = \pb{v}{\zee{}'}$, if $\sum_{u  :\, \text{\normalfont $u$ is an ancestor of $v$}}\inv{u}{i} = o_{q}(m)$ for every $i$ satisfying $\zee{i} \neq \zee{i}'$.
	\end{hypothesis}
	Intuitively, one can understand $\inv{v}{i}$ as the ``amount'' of cost $\algo_v$ pays for querying bits in the $i$th block. Hypothesis~\ref{hypo:main} basically says that when reaching $v$ while evaluating $T$, one ``knows'' nothing about $\zee{i} = \bigoplus_{1 \leq j \leq m} \ex{i}{j}$ if the total amount of cost that has been paid for the $i$th block is not enough.
	If we assume the hypothesis holds, then $T$ naturally induces a classical decision tree that computes $f$, as follows.
	Suppose we want to compute $f(\zee{})$.
	The strategy is to simulate the evaluation process of $T$ while pretending the input is a random variable $\boldsymbol{\exs}$, which is chosen from $\{ \exs : \map(\exs)=\zee{} \}$ uniformly at random.
	Hence, the probability that node $v$ is reached in our simulation will be exactly $\pb{v}{\zee{}}$.
	Initially we do not have any information about $\zee{}$ except the fact that it is a legal input.
	However, we can still simulate the process because $\pb{u}{\zee{}}$ are equal for all possible $\zee{}$ by the hypothesis, where $u$ is any child of the root.
	Once we find ourself having paid enough for the $i$th block during the simulation, we (classically) query $\zee{i}$ to get its value.
	It is easy to see that this query strategy suffices for the entire simulation, i.e., we always know $\pb{u}{\zee{}}$ where $u$ is any child of the current node $v$, despite that we may not have full information of $z$.
	Finally, output the value of the leaf in $T$ we reached.
	The correctness is guaranteed by the fact that $T$ computes $f \circ \func{Xor}$. 
	
	\paragraph{Implementation.}
	Let us consider Hypothesis~\ref{hypo:main} more detailedly.
	First, note that the hypothesis itself does not hold\footnote{In fact, it is easy to construct a counterexample.} because its statement is too definite.
	Nevertheless, we only need it to hold \emph{approximately}.
	For example, the statement $\pb{v}{\zee{}} = \pb{v}{\zee{}'}$ may be relaxed to $\pb{v}{\zee{}} \in (1 \pm 1/\poly(n)) \pb{v}{\zee{}'}$.
	Or, we may allow {few} nodes in $T$ to violate the statement.
	Next, we need to determine $\inv{v}{i}$ for every $v$ and $i$.
	The intuition about $\inv{v}{i}$ makes one think of \emph{query magnitude}, a notion that is first introduced by Bennett et al.~\cite{DBLP:journals/siamcomp/BennettBBV97}.
	Suppose we have a classical decision tree $\algo$ and the input $y$ is determined. 
	It is easy to define $\qm{\algo}{i}{y}$ as the expected number of queries at the $i$th bit when evaluating $\algo(y)$.
	Query magnitude is a quantum analog of $\qm{\algo}{i}{y}$ when $\algo$ is a quantum decision tree.
	We use the same symbol $\qm{\algo}{i}{y}$ to denote the query magnitude at the $i$th bit of $\algo$ on input $y$ hereafter.
	Recall from~\cite{DBLP:journals/siamcomp/BennettBBV97} that $\sum_{i} \qm{\algo}{i}{y}$ is exactly the query complexity of $\algo$, for any $y$. 
	We would like to determine $\inv{v}{i}$ based on $\qm{\algo_v}{i,j}{\exs}$ for all legal input $\exs$.
	This task will be extremely easy when $q=1$ because  $\qm{\algo_v}{i,j}{\exs}$ remains the same for any $\exs$.
	Therefore we can make $\inv{v}{i} = \sum_{1 \leq j \leq m} \qm{\algo_v}{i,j}{\exs}$ for any $x$.
	The case $q>1$ will be discussed in Section~\ref{subsec:simon.disc}.
	Below we assume $q=1$.
	
	For a clearer illustration of the idea, we will assume the hybrid decision tree $T$ in Hypothesis~\ref{hypo:main} satisfies that $\sum_{u  :\, \text{\normalfont ancestor of $v$}}\inv{u}{i} = o(m)$ for any node $v$ and index $i$ in this section.
	Our goal is to prove an approximate version of Hypothesis~\ref{hypo:main}, that $\sum_{\text{$v$ leaf}} | \pb{v}{\zee{}} - \pb{v}{\zee{}'} | \leq 1/\poly(n)$ for any $v$, $\zee{}$ and $\zee{}'$.
	Intuitively, the inequality says that $T$ cannot differentiate between $\boldsymbol{\exs}$ and $\boldsymbol{\exs'}$ with probability $1/2+1/\poly(n)$, where $\boldsymbol{\exs}$ and $\boldsymbol{\exs'}$ are chosen uniformly at random from $\{\exs : \map(\exs)=\zee{}\}$ and $\{\exs' : \map(\exs')=\zee{}'\}$, respectively.
	Let $\pb{v}{\ast} := \expt_{\exs \in \oi^{nm}} [\pb{v}{\exs}]$.
	It suffices for us to consider $|\pb{v}{\zee{}} - \pb{v}{\ast}|$ for every $v$ and $\zee{}$, and then apply the triangle inequality.
	Denote $\oplus \zee{I} := \bigoplus_{i \in I} \zee{i}$.
	By Fourier analysis we have
	\begin{align*}
		|\pb{v}{\zee{}} - \pb{v}{\ast}| =& \left| \expt_{\exs:\, \map(\exs)=\zee{}} [\pb{v}{\exs}] - \expt_{\exs \in \oi^{nm}} [\pb{v}{\exs}] \right| \\
		\leq & \frac12 \sum_{\varnothing \neq I \subseteq \{1,\dots,n\}} \left| \expt_{\exs:\,\oplus(\map(\exs))_I=0}[\pb{v}{\exs}] - \expt_{\exs:\,\oplus(\map(\exs))_I=1}[\pb{v}{\exs}] \right|.
	\end{align*}
	Note that the right-hand side is irrelevant to $\zee{}$.
	For any fixed $I$, we are familiar with the expression $\delta_{I,v} := \left| \expt_{\exs:\,\oplus(\map(\exs))_I=0}[\pb{v}{\exs}] - \expt_{\exs:\,\oplus(\map(\exs))_I=1}[\pb{v}{\exs}] \right|$: it corresponds to the task of computing $\bigoplus_{\substack{i \in I\\1\leq j \leq m}} \ex{i}{j}$.
	More specifically, it can be shown that if we relabel leaves of $T$ appropriately, it would compute $\bigoplus_{\substack{i \in I\\1\leq j \leq m}} \ex{i}{j}$ with probability $\frac12 \left(1 + \sum_{v:\,\text{leaf}}\delta_{I,v} \right)$ when the input is chosen from $\oi^{nm}$ uniformly at random, and this is also the maximal success probability possible. 
	Therefore, if we can prove that the success probability of $T$ computing the parity function induced by $I$ is no more than one half plus an exponentially (in the number of variables involved, $m|I|$) small term, i.e., $\sum_{v:\,\text{leaf}}\delta_{I,v} \leq 2^{-\Omega(m|I|)}$, then,
	\begin{align*}
		\sum_{v:\,\text{leaf}} |\pb{v}{\zee{}} - \pb{v}{\ast}| 
		\leq & \sum_{\varnothing \neq I \subseteq \{1,\dots,n\}} 2^{-\Omega(m|I|)} = \left( 1+2^{-\Omega(m)} \right)^n -1 \\
		=& \left( 1+2^{-\Omega(C\cdot 1 \cdot \log n)} \right)^n -1 = n^{-\Omega(C)}.
	\end{align*}
	Choosing a sufficiently large $C$ we would be done.
	
	\paragraph{Multiplicative adversary method.}
	By the discussion above, we need to show $\sum_{v:\,\text{leaf}}\delta_{I,v} \leq 2^{-\Omega(m|I|)}$, under the condition that  $\sum_{\substack{u  :\, \text{\normalfont ancestor of $v$} \\ i \in I}}\inv{u}{i} = o(m|I|)$ for any node $v$.
	This task makes one reminiscent of the XOR lemma attributed to Lee and Roland~\cite{DBLP:journals/cc/LeeR13}, which is proved by the \emph{multiplicative adversary method}~\cite{DBLP:conf/coco/Spalek08}.
	\begin{theorem}[XOR lemma, weakened version] \label{thm:xorl}
		An $o(k \qua{f})$-query quantum algorithm that computes $\func{Xor}_k \circ f$ must have success probability no more than $1/2+2^{-\Omega(k)}$. 
	\end{theorem}
	We are interested in the special case that $f$ is an $1$-bit identity function, i.e., $\func{Xor}_k \circ f$ becomes $\func{Xor}_k$, in Theorem~\ref{thm:xorl}, because it seems closely related to what we are going to prove.
	Proving Theorem~\ref{thm:xorl} for this special case is much easier.
	We sketch the idea below.
	
	Suppose the quantum decision tree to be analyzed is specified by unitaries $U_0, U_1,\dots$.
	Define $\ket{\psi_y^{(t)}} = U_t \qo{y} U_{t-1} \qo{y} \cdots \qo{y} U_1 \qo{y} U_0 \ketinit$.
	We sometimes omit the superscript ``$\cdot^{(t)}$'' if the value of $t$ is clear from the context.
	We refer to $\rho$ the Gram matrix of $\{ \ket{\psi_y} \}_{y \in \oi^k}$.
	That is, $\rho = \sum_{y,y' \in \oi^k} \braket{\psi_y}{\psi_{y'}} \ket{y} \bra{y'}$.
	Define $\rho^{(t)}$ analogously.
	Set the adversary matrix $\Gamma$ to be 
	$$
		\sum_{y,y' \in \oi^k} 2^{-\dist(y,y')} \ket{y} \bra{y'} = 
		\left[
			\begin{array}{cc}
				1 & -1/2\\
				-1/2 & 1
			\end{array}
		\right]^{\otimes k},
	$$
	where $\dist(y,y')$ denotes the Hamming distance between $y$ and $y'$.
	Let the unit vector $u := \frac{1}{2^{k/2}}\sum_{y \in \oi^k} \ket{y}$ correspond to a uniform input distribution over $\oi^k$.
	The potential function is defined as
	$
		\Phi(t) = u^{\dagger} (\Gamma \circ \rho^{(t)}) u,
	$
	where $\circ$ denotes the Hadamard product\footnote{a.k.a. the Schur product or the entrywise product. Note that it is basis-sensitive and we use basis states as basis here and hereafter.}.
	Note that $\Phi(t) \geq 0$ because $\Gamma \circ \rho^{(t)}\succeq 0$ due to the Schur product theorem.
	A standard argument then goes to show that
	\begin{enumerate}[label=(\roman*)]
		\item $\Phi(0) = 1/2^k$, and
		\item $\Phi(t)/\Phi(t-1) = \bigo(1)$ for any $t \geq 1$, and
		\item if the algorithm makes $t$ queries in total, the success probability of it computing $\func{Xor}_k$ is no more than $1/2+\eta^k\Phi(t)$ for some constant $\eta<2$.
	\end{enumerate}
	Clearly (i), (ii) and (iii) jointly imply that an $o(k)$-query quantum algorithm that computes $\func{Xor}_k$ must have success probability $1/2+2^{-\Omega(k)}$.
	
	Now we prove that $\sum_{v:\,\text{leaf}}\delta_{I,v} \leq 2^{-\Omega(m|I|)}$ in a way similar to the above.
	Fix $I$.
	Set
	$$
		\Gamma = \sum_{\substack{\exs,\exs' \in \oi^{nm} \\ \text{$\forall i \notin I$ $\forall j$, $\ex{i}{j}=\ex{i}{j}'$ }}} 2^{-| \{ (i,j):\, \ex{i}{j} \neq \ex{i}{j}' \} |} \ket{\exs} \bra{\exs'},
	$$
	which is a block-diagonal matrix with each block being 
	$
		\left[
			\begin{array}{cc}
				1 & -1/2\\
				-1/2 & 1
			\end{array}
		\right]^{\otimes m|I|}
	$.
	Let
	$u_v :=  \newline \frac{1}{2^{nm/2}}\sum_{\exs \in \oi^{nm}} \sqrt{\pb{v}{x}} \ket{\exs}$.
	The potential function for node $v$ is defined as $\Psi(v) = u_v^\dagger \Gamma u_v$.
	We are able to show that
	\begin{enumerate}[label=(\roman*)]
		\item $\Psi(r) = 1/2^{m|I|}$ where $r$ is the root of $T$, and
		\item $\sum_{w:\,\text{child of $v$}} \Psi(w)/\Psi(v) = 2^{\bigo(\sum_{i \in I} \inv{v}{i})}$ for any internal node $v$, and
		\item There exists a constant $\eta<2$ such that $\delta_{I,v} = \left| \expt_{\exs:\,\oplus(\map(\exs))_I=0}[\pb{v}{\exs}] - \expt_{\exs:\,\oplus(\map(\exs))_I=1}[\pb{v}{\exs}] \right| \leq \eta^{m|I|} \Psi(v)$ for any node $v$.
	\end{enumerate}
	Therefore (i), (ii) and (iii) jointly imply that $\sum_{v:\,\text{leaf}}\delta_{I,v} \leq 2^{-\Omega(m|I|)}$ since $\sum_{\substack{u  :\, \text{\normalfont ancestor of $v$} \\ i \in I}}\inv{u}{i} = o(m|I|)$ for any node $v$.
	
	\subsection{Proof of Theorem~\ref{thm:simon}} \label{subsec:simon.proof}
	
	First, we prove some basic facts about $\func{BoolSimon}$.
	Keep in mind that $n$ is the input size of $\func{BoolSimon}$ throughout.
	Let $k$ and $K=2^k$, as mentioned in Section~\ref{subsec:res}, be parameters associated with $\func{Simon}$, $\func{LSBSimon}$ and $\func{BoolSimon}$ in Lemma~\ref{lem:lsbs}, \ref{lem:bsimon} and \ref{lem:simonin}.
	Let $\clad{f}$ be the minimal query complexity of a classical algorithm that computes $f$ with probability at least $2/3$, when the input is sampled from all legal inputs uniformly at random.
	The $\cla{\func{Simon}} = \Omega(\sqrt{K})$ lower bound is proved by Yao's minimax principle~\cite{DBLP:conf/focs/Yao77,DBLP:journals/tcs/Vereshchagin98} in~\cite{DBLP:journals/siamcomp/Simon97a}, which in fact gives $\clad{\func{Simon}} = \Omega(\sqrt{K})$.
	
	\begin{lemma} \label{lem:lsbs}
		$\clad{\func{LSBSimon}} = \tomega(\sqrt{K})$.
	\end{lemma}
	\begin{proof}
		Define partial functions $g_i$ by letting $g_i(x)$ be the $i$th least significant bit of $\func{Simon}(x)$.
		Then, $\clad{\func{Simon}} = \bigo( \sum_{1 \leq i \leq k} \clad{g_i} \log k)$ via a simple reduction.
		Note that every $g_i$ is isomorphic to $\func{LSBSimon}$ via a permutation acting on the input.
		Therefore $\clad{\func{LSBSimon}} = \tomega(\frac{\clad{\func{Simon}}}{k \log k}) = \tomega(\sqrt{K})$.
	\end{proof}
	
	\begin{lemma} \label{lem:bsimon} 
		$\qua{\func{BoolSimon}} = \bigo(\log^2 n)$ and $\clad{\func{BoolSimon}} = \tomega(\sqrt{n})$.
	\end{lemma}
	\begin{proof}
		$\clad{\func{BoolSimon}} \geq \clad{\func{LSBSimon}}$ via a simple reduction.
		So $\clad{\func{BoolSimon}} = \tomega(\sqrt{K})$ by Lemma~\ref{lem:lsbs}.
		On the other hand, simulating the quantum algorithm for $\func{Simon}$ gives $\qua{\func{BoolSimon}} \leq k \qua{\func{Simon}} = \bigo(\log^2 K)$.
		Finally, note that $n = Kk$.
	\end{proof}
	
	Define the \emph{(partial) assignment} over \emph{support} $I$ as a vector $\zeee = \left(\zeee_{i}\right)_{i \in I}$.
	Denote $\res{\zeee}{I} := \left(\zeee_i\right)_{i \in I}$, where $I$ is a subset of $\zeee$'s support, to be a partial assignment that is the restriction of $\zeee$ on $I$.
	We say $\zeee'$ \emph{extends} $\zeee$ if $\res{\zeee'}{I} = \zeee$, where $I$ is the support of $\zeee$.
	A string $y = y_1 y_2 \dots y_k$ can be interpreted as a partial assignment $(y_i)_{1 \leq i \leq k}$ when needed.
	An assignment $\zeee$ is called a \emph{certificate} for partial function $f$ if $f(y)$ is invariant for every legal input $y$ that extends $\zeee$.
	The following technical lemma will be used when proving Lemma~\ref{lem:afgood}.
	
	\begin{lemma} \label{lem:simonin}
		Let $\zeee$ be an assignment with support $I$, where $|I| \leq \sqrt{K}$.
		If $\zeee$ is not a certificate for $\func{LSBSimon}$, then
		$K^{|I|} \Pr_{y}[y \text{ \rm extends } \zeee] \leq \poly(K)$, where $y$ is sampled from all legal inputs of $\func{LSBSimon}$ uniformly at random.
	\end{lemma}
	\begin{proof}
		The total number of legal inputs is $(K-1) \cdot \frac{K!}{(K/2)!}$.
		On the other hand, note that all components of $\zeee$ are distinct since otherwise it would be a certificate.
		Therefore, the number of legal inputs $y$ extending $\zeee$ is at most $(K-1) \cdot \frac{(K-|I|)!}{(K/2)!}$.
		Thus $K^{|I|} \Pr_{y}[y \text{ \rm extends } \zeee] \leq \poly(K)$ since $|I| \leq \sqrt{K}$.
	\end{proof}
	
	Assume without loss of generality that there exists an \limit{1} hybrid decision tree $T$ of depth $h \leq \sqrt{n}/\log^{10} n$ that computes $\func{BoolSimon}_n \circ \func{Xor}_{m}$, where $m:=C \log n$.
	We plan to show that $\cla{f} = \bigo(h)$.
	Recall that we use $\algo_v$ to denote the quantum decision tree of node $v$ in $T$.
	Regarding $f$ in Section~\ref{subsec:simon.idea} as $\func{BoolSimon}$, we will still use notations $\exs$, $\zee{}$, $\pb{v}{\exs}$, $\pb{v}{\zee{}}$ and $\sigma$ in the same way as before.
	
	\paragraph{Query magnitude.}
	Let $\algo$ be a quantum decision tree specified by unitaries $U_0, U_1, \dots, U_k$.
	Define \emph{query magnitude}~\cite{DBLP:journals/siamcomp/BennettBBV97} at the $i$th bit of $\algo$ on input $y$, $\qm{\algo}{i}{y} = \sum_{0 \leq t \leq k-1} \mo{P_i \ket{\psi_y^{(t)}}}^2$, where $P_i$ is the projector onto the query register containing index $i$, and $\ket{\psi_y^{(t)}} = U_t \qo{y} \ket{\psi_y^{(t-1)}}$ with $\ket{\psi_y^{(0)}} = U_0 \ketinit$.
	Set $\inv{v}{i} := \sum_{1 \leq j \leq m} \qm{\algo_v}{i,j}{\exs}$ for arbitrary $\exs$.
	
	\paragraph{Potential function.}
	For convenience we denote $\bar{I} = \{1,2,\dots,n\}-I$ for any set $I$ hereafter.
	Let $I$ be any non-empty subset of $\{1,2,\dots,n\}$, $\zeee$ be any partial assignment over support $\bar{I}$. 
	Define
	$
		\pot{I}{\zeee}{v} = (2/\sqrt{3})^{m|I|} u_v^\dagger \Gamma u_v\footnote{The factor $(2/\sqrt{3})^{m|I|}$ is added for Lemma~\ref{lem:leqpot} being in a cleaner form.},
	$
	where
	$$
		\Gamma = \sum_{\substack{\exs, \exs' \in \oi^{nm} \\ \map(\exs),\map(\exs') \text{ extend } \zeee \\ \text{$\forall i \in \bar{I}$ $\forall j$, $\ex{i}{j}=\ex{i}{j}'$ }}} 2^{-| \{ (i,j):\, \ex{i}{j} \neq \ex{i}{j}' \} |} \ket{\exs} \bra{\exs'}
		\text{\quad and \quad}
		u_v = \frac{1}{2^{(n(m-1)+|I|)/2}} \sum_{\substack{\exs \in \oi^{nm} \\ \map(\exs) \text{ extend } \zeee }} \sqrt{\pb{v}{x}} \ket{\exs}.
	$$
	Intuitively, this potential function corresponds to the task of computing $\bigoplus_{\substack{i \in I\\1\leq j \leq m}} \ex{i}{j}$, while the input is chosen from $\{ \exs \in \oi^{nm} : \map(\exs) \text{ extend } \zeee\}$ uniformly at random.
	Let $r$ be the root of $T$.
	By definition, $\pot{I}{\zeee}{r} = (2/\sqrt{3})^{m|I|} \cdot 2^{-m|I|} \leq 1/n^{\Omega(C|I|)}$.
	The following lemma upper bounds the growth of $\pot{I}{\zeee}{v}$.
	
	\begin{lemma} \label{lem:grow}
		For any internal node $v$, $\displaystyle \sum_{w:\,\text{\rm child of $v$}} \pot{I}{\zeee}{w} /  \pot{I}{\zeee}{v} \leq 1+3 \sum_{i \in I} \inv{v}{i}.$
	\end{lemma}
	\begin{proof}
		Define $\ket{\psi_{\exs}}$ to be the final state when running $\algo_v$ on input $\exs$.
		The vector $\left( \sqrt{\pb{w}{\exs}} \right)_{w:\,\text{child of $v$}}$ is exactly $\sqrt{\pb{v}{\exs}} \ket{\psi_{\exs}}$ with each component taking the absolute value.
		Therefore,
		$$
			\frac{1}{(2/\sqrt{3})^{m|I|}} \sum_{w:\,\text{\rm child of $v$}} \pot{I}{\zeee}{w} = \sum_{w} u_w^\dagger \Gamma u_w \leq u_v^\dagger (\Gamma \circ \rho) u_v,
		$$
		where $\rho$ denotes the Gram matrix of $\{ \ket{\psi_{\exs}} \}_{\exs \in \oi^{nm}}$, i.e., $\rho = \sum_{\exs,\exs' \in \oi^{nm}} \braket{\psi_{\exs}}{\psi_{\exs'}} \ket{ \exs} \bra{\exs'}$.
		
		We use a clever observation from Ambainis et al.~\cite{DBLP:conf/coco/AmbainisMRR11} to simplify the proof: any $t$-query quantum algorithm is equivalent, in the sense of the same final Gram matrix, to a $2t$-query quantum algorithm that only makes \emph{computing} queries and \emph{uncomputing} queries.
		The computing and uncomputing queries still correspond to the same unitary $\qo{y}$, where $y$ is the input.
		The difference is we guarantee that the quantum state is a superposition of $\{ \ket{i}\ket{0}\ket{a} \}_{i,a}$ (resp. $\{ \ket{i}\ket{y_i}\ket{a} \}_{i,a}$) right before applying a computing (resp. uncomputing) query.
		
		Specifically, the $1$-query algorithm $\algo_v$ is equivalent to measuring the state $U_2 \qo{\exs} U_1 \qo{\exs} U_0 \ketinit$, where the first and the second $\qo{\exs}$ correspond to uncomputing and computing queries respectively, and $U_1$ corresponds to a CNOT operation: $U_1 \ket{i,j} \ket{b} \ket{w_1,w_2} = \ket{i,j} \ket{b} \ket{w_1 \oplus b,w_2}$.
		Let $\ket{\psi_{\exs}^{(0)}} = U_0 \ketinit$, $\ket{\psi_{\exs}^{(1)}} = U_1 \qo{{\exs}} \ket{\psi_{\exs}^{(0)}}$, $\ket{\psi_{\exs}^{(2)}} =  \qo{{\exs}} \ket{\psi_{\exs}^{(1)}}$, and denote as $\rho^{(t)}$ the corresponding Gram matrix.
		Note that $\rho^{(2)} = \rho$ because $U_2$ is unitary.
		Set $\ket{\psi_{\exs,i,j}^{(t)}} = P_{i,j} \ket{\psi_{\exs}^{(t)}}$, where
		$P_{i,j}$ is the projector onto the query register containing index $(i,j)$ and similarly denote as $\rho^{(t)}_{i,j}$ the corresponding Gram matrix.
		Note that although $\qm{\algo_v}{i,j}{\exs}$ is the query magnitude of $\algo_v$ itself, we still have $\qm{\algo_v}{i,j}{\exs} = \mo{P_{i,j} U_0 \ketinit}^2 =: \beta_{i,j}$ by the way we construct $U_0$, $U_1$ and $U_2$.
		Therefore $\inv{v}{i} = \sum_{1 \leq j \leq m} \beta_{i,j}$.
		 
		By definition it is obvious that $\rho^{(0)} = \sum_{\exs,\exs'} \ket{\exs} \bra{\exs}$ and $\rho^{(0)}_{i,j} = \beta_{i,j} \rho^{(0)}$.
		Since $\ket{\psi_{\exs}^{(1)}}$ is obtained by applying a computing query and then a CNOT gate on $\ket{\psi_{\exs}^{(0)}}$, we have $\rho^{(1)}_{i,j} = \rho^{(0)}_{i,j} \circ \Delta_{i,j}$, where $\Delta_{i,j} := \sum_{\exs,\exs':\, \ex{i}{j} = \ex{i}{j}'} \ket{\exs} \bra{\exs}$.
		Similarly, $\rho^{(1)}_{i,j} = \rho^{(2)}_{i,j} \circ \Delta_{i,j} $ because of the uncomputing query.
		Using the fact that $\Gamma$ is a block-diagonal matrix with each block being 
		$
			\left[
				\begin{array}{cc}
					1 & -1/2\\
					-1/2 & 1
				\end{array}
			\right]^{\otimes m|I|}
		$, it is easy to show that $\Gamma \preceq 2 \Gamma \circ \Delta_{i,j}$ and $\Gamma \circ \Delta_{i,j} \preceq 2 \Gamma$.
		And by the definition of $\Gamma$, $\Gamma = \Gamma \circ \Delta_{i,j}$ if $i \in \bar{I}$.
		Therefore,
		\begin{align*}
			 u_v^\dagger (\Gamma \circ \rho) u_v =& \sum_{\substack{i \in I \\ 1 \leq j \leq m}} u_v^\dagger (\Gamma \circ \rho^{(2)}_{i,j}) u_v + \sum_{\substack{i \in \bar{I} \\ 1 \leq j \leq m}} u_v^\dagger (\Gamma \circ \rho^{(2)}_{i,j}) u_v \\
			\leq & 2 \sum_{\substack{i \in I \\ 1 \leq j \leq m}} u_v^\dagger (\Gamma \circ \rho^{(2)}_{i,j} \circ \Delta_{i,j}) u_v  + \sum_{\substack{i \in \bar{I} \\ 1 \leq j \leq m}} u_v^\dagger (\Gamma \circ \rho^{(2)}_{i,j} \circ \Delta_{i,j}) u_v \\
			= &  2 \sum_{\substack{i \in I \\ 1 \leq j \leq m}} u_v^\dagger (\Gamma \circ \rho^{(0)}_{i,j} \circ \Delta_{i,j}) u_v + \sum_{\substack{i \in \bar{I} \\ 1 \leq j \leq m}} u_v^\dagger (\Gamma \circ \rho^{(0)}_{i,j} \circ \Delta_{i,j}) u_v \\
			\leq &  4 \sum_{\substack{i \in I \\ 1 \leq j \leq m}} u_v^\dagger (\Gamma \circ \rho^{(0)}_{i,j}) u_v + \sum_{\substack{i \in \bar{I} \\ 1 \leq j \leq m}} u_v^\dagger (\Gamma \circ \rho^{(0)}_{i,j}) u_v \\
			= &  4 \sum_{\substack{i \in I \\ 1 \leq j \leq m}} \beta_{i,j} \cdot u_v^\dagger \Gamma u_v + \sum_{\substack{i \in \bar{I} \\ 1 \leq j \leq m}} \beta_{i,j} \cdot u_v^\dagger \Gamma u_v \\
			= & (1+3 \sum_{i \in I} \inv{v}{i}) u_v^\dagger \Gamma u_v = \frac{1}{(2/\sqrt{3})^{m|I|}} (1+3 \sum_{i \in I} \inv{v}{i}) \pot{I}{\zeee}{v}.
		\end{align*}
	\end{proof}
	
	Next, we show that $\pot{I}{\zeee}{v}$ upper bounds $\frac{1}{2} \left| \expt_{\substack{\zee{}:\,\zee{} \text{ \rm extends } \zeee \\ \oplus \zee{I} = 0}}[ \pb{v}{\zee{}}] - \expt_{\substack{\zee{}:\,\zee{} \text{ \rm extends } \zeee \\ \oplus \zee{I} = 1}}[ \pb{v}{\zee{}}] \right| $, which is an important quantity we will use  later in Lemma~\ref{lem:diff}.
	
	\begin{lemma} \label{lem:leqpot}
		For any node $v$, $\left| \displaystyle \expt_{\zee{}:\,\zee{} \text{ \rm extends } \zeee}[(-1)^{\oplus \zee{I}} \pb{v}{\zee{}}] \right| \leq \pot{I}{\zeee}{v}$.
	\end{lemma}
	\begin{proof}
		Note that $\displaystyle \expt_{\zee{}:\,\zee{} \text{ \rm extends } \zeee}[(-1)^{\oplus \zee{I}} \pb{v}{\zee{}}] = u_v^\dagger D u_v$, where $D = \sum_{\exs} (-1)^{\oplus \map(\exs)_I} \ket{\exs} \bra{\exs}$ is a block-diagonal matrix with each block being
		$
			\left[
				\begin{array}{cc}
					1 & 0\\
					0 & -1
				\end{array}
			\right]^{\otimes m|I|}
		$.
		Recall that $\Gamma$ is a block-diagonal matrix with each block being 
		$
			\left[
				\begin{array}{cc}
					1 & -1/2\\
					-1/2 & 1
				\end{array}
			\right]^{\otimes m|I|}
		$.
		By knowledge of quadratic form we have $|u_v^\dagger D u_v| \leq (2/\sqrt{3})^{m|I|} u_v^\dagger \Gamma u_v = \pot{I}{\zeee}{v}$.
	\end{proof}
	
	\paragraph{Approximating $\pb{v}{\zee{}}$.}
	Denote as $J_{\rm min}(v)$ the collection of indices $i$ such that $\sum_{u  :\, \text{\normalfont ancestor of $v$}}\inv{u}{i} \geq m/C = \log n$.
	Recall that $\func{BoolSimon}$ is the Boolean version of $\func{LSBSimon}$.
	So, one position of the input of $\func{LSBSimon}$ corresponds to $\bigo(\log n)$ input bits of $\func{BoolSimon}$.
	We let $J(v)$ consist of those index $i$ such that there exists index $j \in J_{\min}(v)$ corresponding to the same input position of $\func{LSBSimon}$ with $i$.
	To prove an approximate version of Hypothesis~\ref{hypo:main}, for any $v$ and $\zee{}$ we are going to approximate $\pb{v}{\zee{}}$ with $\pb{v}{\zeee{}} := \expt_{\zee{}':\text{ $\zee{}'$ extends $\zeee$}}[\pb{v}{\zee{}'}]$, where $\zeee := \res{\zee{}}{J}$ and $J := J(v)$\footnote{Following the idea in Section~\ref{subsec:simon.idea} one should let $J := J_{\rm min}(v)$, which seems more natural. However, it would bring difficulty in proving Lemma~\ref{lem:afgood} if so, because Claim~\ref{cla:temp2} would no longer hold.}.
	Denote $\diff{\zeee}{v} = \sum_{\varnothing \neq I \subseteq \bar{J}} \expt_{\substack{\zeee':\, \zeee' \text{ \rm has support } \bar{I} \\ \zeee' \text{ \rm extends } \zeee}} [ \pot{I}{\zeee'}{v} ]$.
	The following lemma bounds the difference between $\pb{v}{\zee{}}$ and $\pb{v}{\zeee{}}$.
	
	\begin{lemma} \label{lem:diff}
		For any node $v$, $\displaystyle |\pb{v}{\zee{}} - \pb{v}{\zeee{}} |\leq \diff{\zeee}{v}$.
		
	\end{lemma}
	\begin{proof}
		Let $J := J(v)$. Write $\pb{v}{\zee{}}$, for those $\zee{}$ that extends $\zeee$, in the multilinear (in $(-1)^{\zee{i}}$ for $i \in \bar{J}$) representation:
		$$
			\pb{v}{\zee{}} = \sum_{I \subseteq \bar{J}} c_I \prod_{i \in I} (-1)^{\zee{i}}, \qquad \text{where } c_I = \expt_{\zee{}:\, \zee{} \text{ extends } \zeee} \left[\pb{v}{\zee{}} \prod_{i \in I} (-1)^{\zee{i}}  \right].
		$$
		We have $|\pb{v}{\zee{}} - \pb{v}{\zeee{}} | \leq \sum_{\varnothing \neq I \subseteq \bar{J}} |c_I|$.
		Note that 
		$$
			c_I = \expt_{\zee{}:\, \zee{} \text{ extends } \zeee} [ (-1)^{\oplus \zee{I}} \pb{v}{\zee{}} ] = 
			\expt_{\substack{\zeee':\, \zeee' \text{ has support } \bar{I} \\ \zeee' \text{ extends } \zeee}} \left[ \expt_{\zee{}:\, \zee{} \text{ extends } \zeee'} [ (-1)^{\oplus \zee{I}} \pb{v}{\zee{}} ] \right].
		$$
		Therefore $|c_I| \leq \expt_{\substack{\zeee':\, \zeee' \text{ has support } \bar{I} \\ \zeee' \text{ extends } \zeee}} [ \pot{I}{\zeee'}{v} ]$ by Lemma~\ref{lem:leqpot}.
	\end{proof}
	
	\paragraph{Induced classical algorithm.}
	We now develop a classical algorithm, denoted as $\algo^{\star}$, based on $T$, which is hoped to compute $\func{BoolSimon}(\zee{})$.
	The algorithm runs as follows.
	Define $J^+(v) = J(u)$ for any internal node $v$ with any of its child $u$. 
	We maintain a ``current node'' $v$.
	Initially $v$ is the root of $T$.
	Each time we (i) query $\zee{i}$ (if we still do not know its value) for every index $i \in J^+(v)$, and then (ii) set $v \leftarrow u$ for a child $u$ of $v$, with probability $\pb{u}{\zeee}/\pb{v}{\zeee}$, where $\zeee$ is a partial assignment that contains all queried bits of $\zee{}$.
	The iteration stops in two cases:
	if $v$ is a leaf, then return its value;
	if the bits queried suffices to determine $\func{BoolSimon}(\zee{})$, i.e., they form a certificate, then return it.
	Recall that $h$ denotes the depth of $T$ and that
	we have $\sum_{1 \leq i \leq n} \inv{v}{i} = 1$ for any node $v$, which implies $|J_{\rm min}(v)| \leq h/\log n$, which implies $|J(v)| \leq \bigo(|J_{\rm min}(v)| \log n) \leq \bigo(h)$.
	Thus the query complexity of the algorithm is $\bigo(h)$.
	
	To analyze the behavior of $\algo^{\star}$ we introduce the following notion.
	For any node $v$, let $\zeee$ be a partial assignment over support $J(v)$.
	We say $v$ is \emph{\bad{\zeee}} if (i) $\diff{\zeee}{v} \geq \frac{1}{n^2} \pb{v}{\zeee}$ and (ii) $\zeee$ is not a certificate for $\func{BoolSimon}$ and (iii) none of $v$'s ancestor $u$ is \bad{\zeee'}, where $\zeee' := \res{\zeee}{J(u)}$.
	Intuitively, if on input $\zee{}$, the algorithm $\algo^{\star}$ has not reached any \bad{\zeee} node at the end with high probability (here the randomness comes from $\algo^{\star}$ itself), where $\zeee$ is defined in the description of $\algo^{\star}$, then it will return $\func{BoolSimon}(\zee{})$ with probability about $2/3$ because $T$ computes $\func{BoolSimon} \circ \func{Xor}$ with probability at least $2/3$.
	We formalize this intuition as follows.
	Recall from Section~\ref{subsec:or.proof} that a cut $S$ of $T$ is a set of nodes in $T$ such that any path connecting the root and a leaf passes through exactly one node in $S$.
	
	\begin{lemma} \label{lem:afprob}
		Let $\zee{}$ be a legal input of $\func{BoolSimon}$.
		$\algo^{\star}(\zee{})$ returns $\func{BoolSimon}(\zee{})$ with probability at least $2/3 - \sum_{v:\, \text{\normalfont \bad{\zeee}}} \pb{v}{\zee{}} - \bigo(1/n)$, where $\zeee := \res{\zee{}}{J(v)}$.
	\end{lemma}
	\begin{proof}
		Let $T'$ be the tree induced by the root of $T$, after removing from $T$ nodes $S_{\rm bad}$, which consists of those nodes $v$ such that $v$ is \bad{\zeee} where $\zeee := \res{\zee{}}{J(v)}$.
		And let $T''$ be the tree induced by the root of $T'$, after removing from $T'$ descendants of nodes in $S_{\rm exempt}$, which consists of those internal nodes $v$ such that (i) $\zeee$ is a certificate for $f$ where $\zeee := \res{\zee{}}{J^+(v)}$ and (ii) none of $v$'s ancestor satisfies condition (i).
		Then, we have $\diff{\zeee}{v} \le \frac{1}{n^2} \pb{v}{\zeee}$ for each node $v$ in $T''$.
		Note that this implies $\pb{v}{\zee{}'}/\pb{v}{\zeee} \in \exp(\pm 1/n^2)$ for any $\zee{}'$, by Lemma~\ref{lem:diff}.
		Besides, $S_{\rm bad} \cup S_{\rm exempt} \cup S_{\rm leaf}$ is a cut of $T$, where $S_{\rm leaf}$ consists of those leaves of $T$ in $T''$.
		Define $p_v$ as the probability that $\algo^{\star}$ reaches node $v$ in $T$ on input $\zee{}$.
		
		\begin{claim}
			Let $v$ be any internal node in $T''$ of depth $h$.
			Then, $p_v/\pb{v}{\zee{}} \in \exp(\pm 3h/n^2)$.
		\end{claim}
		\begin{proof}
			Prove by induction.
			Let $v$ be any internal node in $T''$ and $u$ be any of its children.
			Denote by $h$ the depth of $v$.
			Assume $p_v/\pb{v}{\zee{}} \in \exp(\pm 3h/n^2)$. 
			Let $\zeee := \res{\zee{}}{J(v)}$ and $\zeee' := \res{\zee{}}{J(u)}=\res{\zee{}}{J^+(v)}$.
			Now that $\pb{v}{\zee{}'}/\pb{v}{\zeee} \in \exp(\pm 1/n^2)$ for any $\zee{}'$, we have $\pb{v}{\zeee'}/\pb{v}{\zeee} \in \exp(\pm 1/n^2)$ because $\pb{v}{\zeee'} = \expt_{\zee{}':\,\zee{}' \text{ extends } \zeee}[\pb{v}{\zee{}'}]$.
			Therefore,
			$$
				\frac{p_u}{\pb{u}{\zee{}}} = \frac{p_v (\pb{u}{\zeee'}/\pb{v}{\zeee'})}{\pb{u}{\zee{}}} = \frac{p_v}{\pb{v}{\zee{}}} \cdot \frac{\pb{v}{\zee{}}}{\pb{v}{\zeee}} \cdot \frac{\pb{v}{\zeee}}{\pb{v}{\zeee'}} \cdot \frac{\pb{u}{\zeee'}}{\pb{u}{\zee{}}} \in \exp(\pm (3h+3)/n^2).
			$$
		\end{proof}
		By assumption $h \leq n$, we thus obtain $p_v/\pb{v}{\zee{}} \in 1 \pm \bigo(1/n)$ for any node $v$ in $T''$.
		Call a leaf of $T$ \emph{good} if its value is $\func{BoolSimon}(\zee{})$.
		Since $T$ computes $\func{BoolSimon} \circ \func{Xor}$ with probability at least $2/3$, $\sum_{\substack{v:\, \text{leaf of $T$} \\ v \text{ is not good}}} \pb{v}{\zee{}} \leq 1/3$.
		$\algo^{\star}$ returns $\func{BoolSimon}(\zee{})$ when it finally reaches a good leaf or a node in $S_{\rm exempt}$.
		Therefore the success probability is at least 
		\begin{align*}
			\sum_{\substack{v \in S_{\rm leaf} \\ v \text{ is good}}} p_v + \sum_{v \in S_{\rm exempt}} p_v 
			\geq & \left(1 - \bigo\left( \frac1n \right)\right) \left( \sum_{\substack{v \in S_{\rm leaf} \\ v \text{ is good}}} \pb{v}{\zee{}} + \sum_{v \in S_{\rm exempt}} \pb{v}{\zee{}} \right) \\
			= & \left(1 - \bigo\left( \frac1n \right)\right) \left( 1- \sum_{\substack{v \in S_{\rm leaf} \\ v \text{ is not good}}} \pb{v}{\zee{}} - \sum_{v \in S_{\rm bad}} \pb{v}{\zee{}} \right) \\
			\geq &  \left(1 - \bigo\left( \frac1n \right)\right) \left( 1-\frac13 - \sum_{v \in S_{\rm bad}} \pb{v}{\zee{}} \right). 
		\end{align*}
	\end{proof}
	
	\paragraph{Correctness of $\algo^{\star}$.}
	Denote $\zeee = \res{\zee{}}{J(v)}$.
	By Lemma~\ref{lem:diff}, if we can show that $ \sum_{v:\,v \text{ \normalfont is \bad{\zeee}}} \pb{v}{\zee{}}$ is small for all legal input (of $\func{BoolSimon}$) $\zee{}$, then $\algo^{\star}$ computes $\func{BoolSimon}$ with good probability.
	Unfortunately this seems difficult.
	However, it suffices to prove that $\expt_{\zee{}:\, \text{legal}} [ \sum_{v:\, \text{\normalfont \bad{\zeee}}} \pb{v}{\zee{}}]$ is small because $\func{BoolSimon}$ is still hard to compute classically when the input is chosen from $\{\zee{} : \zee{} \text{ is legal}\}$ uniformly at random, i.e., $\clad{\func{BoolSimon}} = \tomega(\sqrt{n})$.
	
	\begin{lemma} \label{lem:afgood}
		$\algo^{\star}$ computes $\func{BoolSimon}$ with probability at least $2/3 - \bigo(1/n)$, when the input is chosen from all legal inputs uniformly at random.
	\end{lemma}
	\begin{proof}
		A node $v$ is $\zeee$-bad implies $n^2 \diff{\zeee}{v} \geq \pb{v}{\zeee}$, which implies $(n^2+1) \diff{\zeee}{v} \geq \pb{v}{\zee{}}$ for any $\zee{}$ due to Lemma~\ref{lem:diff}.
		Hence,
		\begin{align*}
			\expt_{\zee{}:\, \text{legal}} \left[ \sum_{v:\, \text{\normalfont  \bad{\zeee}}} \pb{v}{\zee{}} \right]
			& \leq  (n^2+1) \expt_{\zee{}:\, \text{legal}} \left[ \sum_{v:\, \text{\normalfont  \bad{\zeee}}} \diff{\zeee}{v} \right] \\
			& =   (n^2+1) \expt_{\zee{}:\, \text{legal}} \left[ \sum_{v:\, \text{\normalfont  \bad{\zeee}}} \sum_{\varnothing \neq I \subseteq \overline{J(v)}} \expt_{\substack{\zeee':\, \zeee' \text{ \rm has support } \bar{I} \\ \zeee' \text{ \rm extends } \zeee}} [ \pot{I}{\zeee'}{v} ] \right] \\
			& =   (n^2+1) \expt_{\zee{}:\, \text{legal}} \left[ \sum_{v:\, \text{\normalfont  \bad{\zeee}}} \sum_{\varnothing \neq I \subseteq \overline{J(v)}} \frac{1}{2^{n-|I|-|J(v)|}}\sum_{\substack{\zeee':\, \zeee' \text{ \rm has support } \bar{I} \\ \zeee' \text{ \rm extends } \zeee}} [ \pot{I}{\zeee'}{v} ] \right] \\
			& =   (n^2+1) \sum_{\substack{\varnothing \neq I \subseteq \{1,\dots,n\} \\ \zeee':\, \zeee' \text{ \rm has support } \bar{I} \\ v:\, \text{\normalfont  \bad{\zeee}, } \zeee := \res{\zeee'}{J(v)} \\  I \subseteq \overline{J(v)}} } \left( \pot{I}{\zeee'}{v}  \cdot \frac{1}{2^{n-|I|-|J(v)|}} \Pr_{\zee{}:\, \text{legal}}[\text{$z$ extends $\zeee$}] \right).
		\end{align*}
		The right-hand side of the formula above is a linear combination of $ \pot{I}{\zeee'}{v}$.
		$\pot{I}{\zeee'}{v}$ makes contribution only if $v$ is \bad{\zeee} with $\zeee := \res{\zeee'}{J(v)}$.
		By definition of \bad{\zeee}, those $v$, such that $\pot{I}{\zeee'}{v}$ makes contribution, form a subset of a cut of $T$, i.e., there do not exist two nodes such that one node is the ancestor of the other one.
		Furthermore, these $v$ satisfy that $\overline{J_{\min}(v)} \supseteq \overline{J(v)} \supseteq I$.
		Finally, the coefficient of term $\pot{I}{\zeee'}{v}$ is no more than $\poly(n)$ by Claim~\ref{cla:temp2}.
		Along with Claim~\ref{cla:temp1}, we get 
		$$ 
			\expt_{\zee{}:\, \text{legal}} \left[ \sum_{v:\, \text{\normalfont \bad{\zeee}}} \pb{v}{\zee{}} \right] \leq (n^2+1) \sum_{I,\zeee'} 1/n^{\Omega(C|I|)} \cdot \poly(n) \leq 1/n^{\Omega(C)},
		$$
		if $C$ is sufficiently large.
		As the result, we are done due to Lemma~\ref{lem:afprob}.
		
		\begin{claim} \label{cla:temp1}
			Let $S$ be a cut of $T$ such that $\overline{J_{\min}(v)} \supseteq I$ for all $v\in S$.
			Then, $\sum_{v \in S} \pot{I}{\zeee}{v} = 1/n^{\Omega(C|I|)}$.
		\end{claim}
		\begin{claim} \label{cla:temp2}
			Let $\zeee := \res{\zeee'}{J(v)}$.
			Then, $2^{|J(v)|} \Pr_{\zee{}:\, \text{\normalfont legal}}[\text{\normalfont $z$ extends $\zeee$}] \leq \poly(n)$.
		\end{claim}
		\begin{proof}[Proof of Claim~\ref{cla:temp1}]
			Let $T'$ be the tree induced by those nodes $v$ of $T$ such that $\overline{J(v)} \supseteq I$.
			Define $\pott{v} := \pot{I}{\zeee}{v} \exp \left( 3|I| \log n - 3 \sum_{\substack{i \in I \\ u:\, \text{ancestor of $v$}}} \inv{u}{i} \right)$.
			Let $v$ be any node in $T'$.
			$\pott{v} \geq \pot{I}{\zeee}{v}$ because $\overline{J_{\min}(v)} \supseteq I$.
			On the other hand, $\sum_{u:\, \text{child of $v$}} \pott{u} \leq \pott{v}$ by Lemma~\ref{lem:grow}.
			Therefore $\sum_{v \in S} \pot{I}{\zeee}{v} \leq \sum_{v \in S} \pott{v} \leq \pott{r} = \pot{I}{\zeee}{r} \exp(3|I| \log n) = 1/n^{\Omega(C|I|)}$, where $r$ is the root of $T$.
		\end{proof}
		\begin{proof}[Proof of Claim~\ref{cla:temp2}]
			$\func{BoolSimon}$ is the Boolean version of $\func{LSBSimon}$.
			Therefore, $\zeee$ corresponds to a partial assignment on the input of $\func{LSBSimon}$, because of the way we construct $J(v)$.
			Then the claim follows by Lemma~\ref{lem:simonin}.
		\end{proof}
	\end{proof}
	
	As discussed before, the complexity of $\algo^{\star}$ is $\bigo(h)$.
	Combining Lemma~\ref{lem:afgood} and Lemma~\ref{lem:bsimon} we have $\clad{\func{BoolSimon}} \in \bigo(h) \cap \tomega(\sqrt{n})$.
	Therefore $h = \tomega(\sqrt{n})$.
	
	\subsection{Discussion} \label{subsec:simon.disc}
	
	The proof in the previous section settles a special case of Conjecture~\ref{conj:main}, that $f = \func{BoolSimon}$ and $q = 1$.
	Here we discuss the possibility of adapting the proof to more general cases, and problems one may face if doing so.
	
	\paragraph{General $f$.}
	Let us see what would happen if one directly apply the proof for an arbitrary partial Boolean function $f$.
	First, $J:=J_{\rm min}(v)$ seems to be the choice when approximating $\pb{v}{\zee{}}$ in Lemma~\ref{lem:diff} because we no longer know any structure of $f$ a priori.
	Then, everything would go smoothly except Lemma~\ref{lem:afgood} or more exactly, Claim~\ref{cla:temp2}.
	The claim holds originally because $\func{BoolSimon}$ is the Boolean version of $\func{LSBSimon}$, which has a nice property that before a classical decision tree can determine $\func{LSBSimon}(y)$ with good probability, the known information about $y$ looks like if \emph{$y$ is just drawn from $\{0,\dots,K-1\}^K$ uniformly at random} (Lemma~\ref{lem:simonin}).\footnote{We suspect that $\func{For}$ or its variant may have this property as well. If so, the proof directly applies for the case $f = \func{For}$.}
	
	There do exist a way to adapt the proof so that we can handle more cases such as $f = \func{Or}$, but it seems still not enough for every $f$.
	The modification is that one first identify the hard distribution of $f$, in the sense that $f$ is still hard to compute classically even the input is sampled from this distribution, via Yao's minimax principle~\cite{DBLP:conf/focs/Yao77,DBLP:journals/tcs/Vereshchagin98}.
	\newcommand{\bz}{\boldsymbol{\zee{}}}
	Then, redefine $\pb{v}{\zeee} := \sum_{\zee{}} \pb{v}{z} \Pr_{\bz}[\bz = z \;|\; \bz \text{ extends } \zeee]$, where $\bz$ is a random variable sampled from the hard distribution.
	Lemma~\ref{lem:diff} should be also modified such that the term on the right-hand side should depend on $\zee{}$, i.e., we need a more fine-grained analysis that cares about each $\zee{}$ individually.
	The concept \bad{\zeee} should be upgraded to \bad{(\zeee,\zee{})}.
	And all other relevant details should be modified accordingly.
	
	\paragraph{Larger $q$.}
	Arguably we are more interested in this one, because in order to prove an oracle separation version of Conjecture~\ref{conj:ours}, it suffices to set $f = \func{BoolSimon}$ and $q = C \log^k n$ in Conjecture~\ref{conj:main}.
	Suppose one is going to directly generalize the proof for some $q>1$.
	The first issue is to redetermine $\inv{v}{i}$ since this task is no longer that trivial as we did in Section~\ref{subsec:simon.idea}.
	We believe this is doable for $q=\bigo(1)$, e.g., one may set $\inv{v}{i}$ to be some weighted average of $\sum_{1 \leq j \leq m} \qm{\algo_v}{i,j}{\exs}$ and use a more complicated analysis instead of Lemma~\ref{lem:grow}.
	
	However, there is a technical barrier for the current method when we study non-constant $q = q(n)$.
	The culprit is the term $o_q(m)$ in Hypothesis~\ref{hypo:main}: if it is actually, say, $o(m/2^q)$, then we cannot prove a bound better than
	$
		\hyb{f_n \circ \func{Xor}_{C q \log n}}{q} = \Omega(\cla{f_n} \cdot q \log n / 2^q).
	$
	As a result, the bound becomes completely trivial when $q \gg \log n$.
	
	Unfortunately, this is indeed the case.
	Let $q < \log n$.
	Suppose $T$ is a hybrid decision tree, in which each depth-$h$ node holds an algorithm that, on input $\exs$ (of size $nm$), (i) queries the first $q-1$ bits (in the $1$st block) to get a number $0 \leq i \leq 2^{q-1}-1$ and then (ii) queries $\ex{i+2}{h}$.
	It follows that $T$ will ``know'' of $\zee{i} = \bigoplus_{1 \leq j \leq m}\ex{i}{j}$ at depth $m$, for some $i$ dependent on the $1$st block.
	However, it seems that any ``reasonable'' assignment of $\inv{v}{i}$ should satisfy that $\inv{v}{2}=\inv{v}{3}=\dots=\inv{v}{2^{q-1}+1}=\bigo(1/2^q)$.
	Therefore, even an approximate version of Hypothesis~\ref{hypo:main} would be incorrect if we replace the term $o_q(m)$ by, say, $o(m/(1.9)^q)$.
	Hence at least a new technical ingredient is necessary.
	
	\section*{Acknowledgement}
	We thank Jiaqing Jiang, Yun Sun and Zhiyu Xia for helpful discussions.
		
	
	\bibliographystyle{alpha}
	\bibliography{reference}

\end{document}